\documentclass{article}

\usepackage[ruled,linesnumbered,noend]{algorithm2e}
\usepackage{amsmath,amsfonts,amssymb,amsthm}
\usepackage{authblk}
\usepackage{bbm}
\usepackage{booktabs}
\usepackage{enumerate}
\usepackage{graphicx,psfrag,epsf}
\usepackage[utf8]{inputenc}
\usepackage{mathtools}
\usepackage{natbib}
\usepackage{pgfplots}
\usepackage{physics}
\usepackage{subcaption}
\usepackage{tikz}
\usepackage{xcolor}
\usepackage[breaklinks=true]{hyperref}

\usepackage{geometry}

\DeclareMathOperator{\argmin}{argmin}

\newcommand{\jasa}{0}

\SetKwInOut{Ret}{Return}
\SetKwProg{Fn}{Function}{:}{}
\SetKwFunction{RC}{ReflectionCoupling}
\SetKwFunction{RejC}{RejectionCoupling}
\SetKwFunction{ERC}{EnsembleRejectionCoupling}
\SetKwFunction{RejR}{CoupledRejectionResampling}
\SetKwFunction{Thorisson}{ThorissonCoupling}

\newtheorem*{corollary}{Corollary}
\newtheorem{definition}{Definition}

\newtheorem{remark}{Remark}
\newtheorem{proposition}{Proposition}

\usetikzlibrary{automata}
\usetikzlibrary{positioning}  %
\usetikzlibrary{arrows}       %

\usepgfplotslibrary{groupplots}

\pgfplotsset{
  grid style = {
    dash pattern = on 0.25mm off 0.75mm,
    line cap = round,
    gray,
    line width = 0.1pt
  }
}

\title{The Coupled Rejection Sampler}

\author[1]{Adrien Corenflos}
\author[1]{Simo S\"{a}rkk\"{a}}
\affil[1]{Department of Electrical Engineering and Automation, Aalto University.}
\date{\today}

\begin{document}
\maketitle

\begin{abstract}
    We propose a coupled rejection-sampling method for sampling from couplings of arbitrary distributions. The method relies on accepting or rejecting coupled samples coming from dominating marginals. Contrary to existing acceptance-rejection coupling methods, the variance of the execution time of the proposed method is limited and stays finite as the two target marginals approach each other in the sense of the total variation norm. In the important special case of coupling multivariate Gaussians with different means and covariances, we derive positive lower bounds for the resulting coupling probability of our algorithm, and we then show how the coupling method can be optimized in closed form. Finally, we show how we can modify the coupled rejection-sampling method to propose from coupled ensemble of proposals, so as to asymptotically recover a maximal coupling. We then apply the method to the problem of coupling rare events samplers, derive a parallel coupled resampling algorithm to use in particle filtering, and show how the coupled rejection-sampler can be used to speed up unbiased MCMC methods based on couplings.
\end{abstract}


\section{Introduction}
\label{sec:intro}
Given two probability distributions $p$ and $q$, a coupling \citep[see, e.g., ][]{Thorisson2000Coupling,Lindvall2002lectures} of $p$ and $q$ is defined as a pair of random variables $(X, Y)$, defined on a joint probability space, that are marginally distributed according to $p$ and $q$, respectively. Typical examples of such couplings, exhibiting different uses and properties, comprise (i) the trivial independent coupling, where $X \sim p$ and $Y \sim q$ with $X$ independent of $Y$; (ii) optimal transport couplings, for which the joint distribution $\Gamma_{OT}$ of $(X, Y)$ minimizes a given objective over all possible couplings $\Gamma_{OT} = \argmin_{\Gamma} \int c(x, y) \Gamma(\dd{x}, \dd{y})$,~\citep{villani2009optimal}; (iii) common random number couplings, where the same uniform/random seed is used to sample the marginals, e.g., if we know the quantile functions $F^{-}$ and $G^{-}$ of $p$ and $q$, we can sample $X = F^{-}(U)$, $Y = G^{-}(U)$ where $U$ is uniformly distributed~\citep{whitt1976bivariate, gal1984optimality}; (iv) similarly to common random numbers, antithetic variables, which are constructed in the same way as the common random numbers couplings but by using $X = F^{-}(U)$, $Y = G^{-}(1-U)$ instead~\citep{hammersley1956new}.

In this article, we consider another class of couplings, which preserve mass over the diagonal of the joint distribution, that is, couplings $(X, Y)$ such that $\mathbb{P}(X=Y) > 0$. Such couplings are ubiquitous in proving a number of inequalities used, for example, to study the convergence of Monte Carlo algorithms~\citep[see, e.g.,][]{Reutter1995general,Lindvall2002lectures}. More recently, they have also have been used as a computational tool rather than a theoretical one. In particular, the seminal paper of \citet{glynn2014exact} showed how to compute unbiased estimates of expectations using coupled Markov chains, allowing to then compute these to an arbitrary precision using distributed hardware~\citep{Jacob2020UnbiasedMCMC}. Following this work, several extensions have been developed, both in the classical Markov chain Monte Carlo (MCMC) methods~\citep{Jacob2020UnbiasedMCMC,Heng2019unbiasedHMC, Xu2021couplings, Wang2021maximal}, as well as in pseudo-marginal and particle Markov chain Monte Carlo (PMCMC) methods~\citep{Jacob2020UnbiasedSmoothing, Middleton2019unbiased, Middleton2020unbiased}. 
    
If $X_t$ and $Y_t$ are two Markov chains, the coupling (or meeting) time \citep[e.g.][]{Thorisson2000Coupling,Lindvall2002lectures} is defined as $\inf\{t \ge 0 \mid X_t = Y_t\}$. In the context of MCMC methods, a common requirement of debiasing coupling methods is that the constructed coupled Markov chains have a coupling time with finite expectation $\mathbb{E}[\tau] < \infty$ \citep[e.g.][]{Jacob2020UnbiasedMCMC} (more precisely, it needs to have polynomial tails~\citep{Middleton2019unbiased}). To construct such Markov chains, one typically relies on implementing a maximum coupling of the proposal distributions $q(\cdot \mid X_t)$, $q(\cdot \mid Y_t)$ or of the Metropolis-Hastings kernel directly~\citep{Wang2021maximal}. 

When considering coupling the proposals, a maximal coupling is a pair of random variables $(X^*, Y^*)$ which are marginally distributed according to $q(\cdot \mid X_t)$ and $q(\cdot \mid Y_t)$, respectively, such that $\mathbb{P}(X^* = Y^*)$ is maximized over the (infinite) set of all possible couplings. When $q(\cdot \mid X_t)$ and $q(\cdot \mid Y_t)$ are Gaussian distributions with the same covariance, this can be done, for example, using the reflection-maximal coupling introduced in \citet{Bou2020coupling}. However, in the case when the covariances are different, one needs to resort to using a rejection sampling algorithm~\cite[see, e.g.][]{Thorisson2000Coupling, Valen1998Diagnosis}. An important drawback of the latter algorithm is that the variance of its (random) run time goes to infinity as the total variation distance of the marginal distributions goes to zero \citep{Gerber2020discussion}.
This property limits its application in contractive schemes, such as unbiased MCMC methods, where the aim is to get the marginals to converge to one another. However, as pointed out in \citet{Gerber2020discussion}, if one is willing to compromise on the maximal property, that is, sample $X^*$ and $Y^*$ with a suboptimal diagonal mass $\mathbb{P}(X^* = Y^*)$, then it is possible to control the variance of the run time. This can be used as a trade-off between computational cost and optimality \citep[][Ch. 2]{Jacob2021lecture}.

The first contribution of this article is to present an algorithm for sampling from couplings of arbitrary distributions by using an acceptance-rejection scheme on proposal samples coming from a dominating coupling. The algorithm has the advantage that the variance of the execution time stays finite as the coupling marginals approach each other. The second contribution is to specialize this to Gaussians with different means and covariances, derive a principled convex optimisation scheme to find a good dominating proposal, and show that the coupling probability is lower-bounded by a positive (parametric) constant. The lower bounds we derive for the acceptance probability also lead to upper bounds for the total variation distance between Gaussian distributions that are always less than $1$. The third contribution is to modify the acceptance-rejection scheme by considering a coupled ensemble proposal such that, as the size of the proposal ensemble increases, the resulting algorithm asymptotically produces samples from a maximal coupling. We finally apply the methods introduced to derive a novel parallel coupled resampling algorithm for particle filtering, and show the attractiveness of our algorithms on several sampling problems.

The structure of the article is the following. In Section~\ref{sec:rejection-coupling} we describe the coupling algorithm for generic distributions, prove its correctness, and analyse its theoretical properties. In Section~\ref{sec:ensemble-rejection}, by extending \citet{Deligiannidis2020ensemble}, we subsequently derive a coupled ensemble rejection sampler that extends the coupled rejection-sampler algorithm designed in Section~\ref{sec:rejection-coupling}. We derive asymptotic bounds for the coupling success probability of the resulting coupling that show how this new method asymptotically recovers a maximal coupling. In Section~\ref{sec:gaussian} we then discuss the special case of designing efficient coupling of Gaussian distributions with different means and covariances. We derive tight positive lower bounds for the coupling success probability and show how to design and compute efficient dominating Gaussian couplings. Finally, in Section~\ref{sec:applications}, we (i) show how the method can be used to sample from couplings of Gaussian tails, (ii) implement a coupled parallel resampling algorithm, (iii) compare our method with the algorithm discussed in \citet{Gerber2020discussion} on a Gibbs sampling problem, and finally (iv) compare our algorithm to \citet{Wang2021maximal} on a Bayesian logistic regression problem using a manifold Metropolis-adjusted Langevin algorithm sampler.

\subsection*{Assumptions and Notations}
    In this article we suppose that all the random variables are defined on a common probability space $(\Omega, \mathcal{F}, \mathbb{P})$ and take value in a measurable space $\chi$ being either $\mathbb{Z}$ or $\mathbb{R}^d$ for some $d \geq 1$. It is, however, worth noting that the results and algorithms of this article directly extend to any measurable space for which the event $\{X=Y\}$ is measurable for all random variables $X, Y$. 
    
    When considering that $X$ is sampled from a probability distribution $p$ we write $X \sim p$. For the sake of notational simplicity, we will assume throughout this article that all distributions, apart from coupling distributions, have a density with respect to a common measure (e.g., Lebesgue measure for densities defined on $\mathbb{R}^d$ and the counting measure on the set of integers $\mathbb{Z}$), and, where no confusion is possible, we will identify the probability measure with its density. We, however, emphasize that all of our results can easily be rewritten in terms of Radon--Nikodym derivatives. Moreover, for the sake of conciseness, we use a slight abuse of notation by denoting $\int (\cdot) \dd{x} = \int (\cdot) \dd{\mu}$ for an integral with respect to an implicit unnamed reference measure $\mu$.
    
    The uniform distribution on $[0,1]$ is denoted by $\mathcal{U}(0, 1)$. The probability distribution of a general coupling of $(X,Y)$ is denoted by $\Gamma$ and the probability distribution of a given dominating coupling is denoted by $\hat{\Gamma}$. When there is no ambiguity we interchangeably refer to the pair $(X, Y)$ or its distribution as the coupling. Finally, for any random variable $X$ and measurable event $E$, we write $X \mid E$ for a random variable with law $\mathbb{P}((X \mid E) \in A) = \mathbb{P}(X \in A \mid E)$.

    When dealing with sequences $u_N$ and $v_N$, we say that they are equivalent and write $u_N \sim v_N$ if we have $\frac{u_N}{v_N} \to 1$ as $N$ goes to infinity. Furthermore, we write $\limsup_{N \to \infty} u_N$ for the quantity $\lim_{N \to \infty} \sup_{K \geq N} u_K$ when it exists, and similarly for $\liminf$.
    
    We denote the Gaussian distribution with mean $\mu$ and covariance $\Sigma$ by $\mathcal{N}(\mu, \Sigma)$ and the corresponding density by $\mathcal{N}(x; \mu, \Sigma)$. For a given covariance matrix $\Sigma$ we write $\Sigma^{1/2}$ for its lower Cholesky factor and $\Sigma^{-1/2}$ for the inverse thereof.

\section{The coupled rejection-sampling method}
\label{sec:rejection-coupling}
Let $p$ and $q$ be two probability densities defined on $\chi$. We define a diagonal coupling of $p$ and $q$ in the following way.
\begin{definition}[Diagonal coupling]
\label{def:diagonal-coupling}
    A pair of random variables $(X,Y)$ defined on the same probability space and taking values in $\chi$ is a diagonal coupling of the probability densities $p$ and $q$ if and only if it is a coupling, that is, if
    \begin{align*}
        \mathbb{P}(X \in A) = \int_A p(x) \dd{x}, \quad \mathbb{P}(Y \in A) = \int_A q(x) \dd{x},
    \end{align*}
    for all measurable sets $A \subset \chi$, and if $\mathbb{P}(X = Y) > 0.$
\end{definition}
With this definition, a maximal coupling \citep{Thorisson2000Coupling, Lindvall2002lectures} is a diagonal coupling with maximal mass on the diagonal event $\{X=Y\}$. The maximal coupling also has a connection with the total variation distance between the distributions $\mathbb{P}(X)$ and $\mathbb{P}(Y)$. If we define total variation distance as $\| p - q \|_{TV} = \sup_A P(X \in A) - P(Y \in A) = 1 - \int \min(p(x), q(x)) \, dx$ then for a maximal coupling $(X, Y)$, we have
\begin{equation}
    \mathbb{P}(X = Y) = 1 - \| p - q \|_{TV} = \int \min(p(x), q(x)) \, dx.
\label{eq:tv_connection}
\end{equation}

In order to construct the acceptance-rejection scheme, we need the following definition.

\begin{definition}[Dominating pair]
\label{def:dominating-pair}
    Let $(\hat{p}, \hat{q})$ be a pair of probability densities. We say that it dominates $(p, q)$ (or that $(p, q) \preceq (\hat{p}, \hat{q})$) if and only if there exists $1 \leq M(p, \hat{p}) <\infty $ and $1 \leq M(q, \hat{q})< \infty$ such that for any $x \in \mathbb{R}^d$ we have
    \begin{align}
        p(x) \leq M(p, \hat{p}) \, \hat{p}(x), \quad q(x) \leq M(q, \hat{q}) \, \hat{q}(x).
    \end{align}
\end{definition}
It is worth noting that even if $(\hat{p}, \hat{q})$ dominates $(p, q)$, a coupling $\hat{\Gamma}$ of $(\hat{p}, \hat{q})$ does not necessarily dominate a coupling $\Gamma$ of $(p, q)$. Indeed, there might not exist $M(\Gamma, \hat{\Gamma})$ such that $\frac{\dd{\Gamma}}{\dd{\hat{\Gamma}}} \leq M(\Gamma, \hat{\Gamma})$ uniformly. However, in a slight abuse of language and notation, we say that a coupling $\hat{\Gamma}$ dominates the coupling $\Gamma$ (i.e., $\hat{\Gamma} \succeq \Gamma$) when its marginals do. 

Now, suppose that we have a diagonal coupling $\hat{\Gamma}$ of $\hat{p}$ and $\hat{q}$ dominating the independent coupling $(p \otimes q)(x, y) = p(x) \, q(y)$. It turns out that we can produce samples from a diagonal coupling $\Gamma$ with marginals $p$ and $q$ by using an acceptance-rejection scheme with the coupling $\hat{\Gamma}$ as the proposal. Algorithm~\ref{alg:rejection-coupling} shows how do this.
\begin{algorithm}[!htb]
    \caption{Rejection-coupling of $(p, q)$}
    \label{alg:rejection-coupling}
    \DontPrintSemicolon
    \Fn{\RejC{$\hat{\Gamma}$, $p$, $q$}}{
        \tcp{Supposing $\hat{\Gamma} \succeq p \otimes q$ is a coupling of $\hat{p}$ and $\hat{q}$.}
        Set $A_X = 0$ and $A_Y = 0$ \tcp{Acceptance flags}
        \While{
            $A_X = 0$ and $A_Y = 0$
        }
        {
            Sample $X_1, Y_1 \sim \hat{\Gamma}$, $U \sim \mathcal{U}(0, 1)$\;
            \lIf{
                $U < \frac{p(X_1)}{M(p,\hat{p})\hat{p}(X_1)}$
            }{
                set $A_X = 1$
            }
            \lIf{
                $U < \frac{q(Y_1)}{M(q,\hat{q})\hat{q}(Y_1)}$
            }{
                set $A_Y = 1$
            }
        }
        Sample $X_2, Y_2$ from $p \otimes q$\;
        \Ret{$X = A_X \, X_1 + (1 - A_X) \, X_2$, $Y = A_Y \, Y_1 + (1 - A_Y) \, Y_2$}
    }
\end{algorithm}

Let us denote by $\tau$ the number of times $\hat{\Gamma}$ is sampled from. We call this quantity the number of steps of Algorithm~\ref{alg:rejection-coupling} as it is roughly proportional to the run time of the algorithm, noting that in reality, sampling from $p \otimes q$ may be expensive too.
The following proposition shows that Algorithm~\ref{alg:rejection-coupling} samples from a coupling of $p$ and $q$, and the expectation and variance of its number of steps are bounded.
\begin{proposition}
\label{prop:rejection-coupling}
    Suppose that $\hat{\Gamma} \succeq p \otimes q$. Then the random variables $X$ and $Y$ generated by Algorithm~\ref{alg:rejection-coupling} are marginally distributed according to $p$ and $q$, respectively, and the mean and variance of the random number of steps $\tau$ of Algorithm~\ref{alg:rejection-coupling} verify
    \begin{equation}
    \begin{split}
        \mathbb{E}[\tau] &\leq \min\left(M(p, \hat{p}), M(q, \hat{q})\right), \\
        \mathbb{V}[\tau] &\leq \min\left(M(p, \hat{p}), M(q, \hat{q})\right)^2 - \min\left(M(p, \hat{p}), M(q, \hat{q})\right).
    \end{split}
    \end{equation}
    Finally, when the coupling for $X_2, Y_2$ is such that $\mathbb{P}(X_2 = Y_2) = 0$, we also have $\mathbb{P}(X = Y) \leq \mathbb{P}(X_1 = Y_1)$.
\end{proposition}
\begin{proof}
    The proof, mostly on classical rejection sampling theory, can be found in {\if1\jasa the supplementary material\else Appendix~\ref{app:proof_prop_rej}\fi}.
\end{proof}

\begin{remark}
    The upper bounds on the mean and variance of number of steps $\tau$ reduce to $1$ and $0$ when the proposals recover the targets.
\end{remark}

\begin{proposition}
\label{prop:diagonal-rejection-coupling}
    Suppose $\hat{\Gamma} \succeq p \otimes q$ is a diagonal coupling of $\hat{p}$ and $\hat{q}$. Furthermore, suppose that for all $x \in \mathbb{R}^d$, we have $\hat{p}(x) > 0 \Rightarrow p(x) > 0$ and similarly for $\hat{q}$ and $q$. In this case, Algorithm~\ref{alg:rejection-coupling} defines a diagonal coupling of $p$ and $q$.
\end{proposition}
\begin{proof}
    Proposition~\ref{prop:rejection-coupling} ensures that Algorithm~\ref{alg:rejection-coupling} defines a coupling of $p$ and $q$. Now, using the notations of Algorithm~\ref{alg:rejection-coupling},
    \begin{equation}
    \begin{split}
        \mathbb{P}(X = Y) 
            &\geq   \mathbb{P}(X_1 = Y_1 \,\&\, A_X = A_Y = 1) \\ 
            &=      \mathbb{P}(A_X = A_Y = 1 \mid X_1 = Y_1) \mathbb{P}(X_1 = Y_1) \\
            &=      \mathbb{P}(X_1 = Y_1) \int c(x) \min\left(\frac{p(x)}{M(p, \hat{p}) \hat{p}(x)}, \frac{q(x)}{M(q, \hat{q}) \hat{q}(x)}\right) \dd{x} ,
    \end{split}
    \end{equation}

    where $c(x)$ is the density of $X_1 \mid X_1 = Y_1$.
    
    Because $\mathbb{P}(X_1 = Y_1) > 0$, we must have a measurable set $A$ with $c(A) > 0$ and, for any $x \in A$, $\hat{p}(x) > 0$ and $\hat{q}(x) > 0$. This ensures that $\int c(x) \min\left(\frac{p(x)}{M(p, \hat{p}) \hat{p}(x)}, \frac{q(x)}{M(q, \hat{q}) \hat{q}(x)}\right) \dd{x} > 0$,
    and finally that $\mathbb{P}(X=Y) > 0$.
\end{proof}

Proposition~\ref{prop:diagonal-rejection-coupling} provides weak conditions for Algorithm~\ref{alg:rejection-coupling} to return a diagonal coupling of $p$ and $q$, however, in practice, the coupling success probability will depend greatly on $\mathbb{P}(A_X = A_Y = 1 \mid X_1 = Y_1)$. This means $M(p, \hat{p})$ and $M(q, \hat{q})$ need to be as close as possible to $1$ for the rejection step to not degrade the coupling success probability too much. This, unfortunately, may come at the cost of making the dominating coupling less successful, and it is not clear what the best trade-off is in this respect. In the next section we discuss how we can completely bypass this trade-off at the cost of additional (parallelisable) computations, while in Section~\ref{sec:gaussian} we will see how we can construct a tight dominating maximal coupling for arbitrary Gaussian distributions at low computational cost.

\section{The ensemble coupled rejection-sampling method}
\label{sec:ensemble-rejection}
As mentioned in Section~\ref{sec:rejection-coupling}, the success probability of the coupling defined by Algorithm~\ref{alg:rejection-coupling} is highly dependent on the acceptance probability of the marginal rejection subroutine. In fact, if the latter goes to $0$ (which is typically the case when the dimension of the sampling space increases), the former will too. Thankfully, in \cite{Deligiannidis2020ensemble}, the authors introduce an ensemble rejection sampler which increases the acceptance probability of a rejection-sampling algorithm by considering not one, but several proposals. Formally, if $p_{\mathrm{RS}}$ is the probability of accepting a sample in the classical rejection-sampling algorithm, then their ensemble rejection-sampler with $N$ proposals has probability $p_{\mathrm{ERS}} \geq \frac{N p_{\mathrm{RS}}}{1 + (N-1) p_{\mathrm{RS}}}$ to result in an accepted sample \citep[][Proposition 1]{Deligiannidis2020ensemble}. 

In order to do so, the authors form the self-normalized importance sample $\sum_{i=1}^N W_i \delta_{X_i}$, where $X_i \sim \hat{p}, i=1, \ldots, N$ are i.i.d., $w_i = \frac{p(X_i)}{\hat{p}(X_i)}$, $W_i = \frac{w_i}{\sum_{j=1}^N w_i}$, for $i = 1, \ldots, N$. When this is done, they sample $I$ from a categorical distribution $\mathrm{Cat}((W_i)_{i=1}^N)$, and then accept the proposal $X_I$ with probability $\frac{\sum_{i=1}^N w_i}{M(p, \hat{p}) + \sum_{i\neq I} w_i}$.

In order to improve the acceptance rate of Algorithm~\ref{alg:rejection-coupling}, we could decide to replace the rejection-sampling step by the ensemble rejection-sampling one and Propositions~\ref{prop:rejection-coupling} and \ref{prop:diagonal-rejection-coupling} would essentially remain unchanged. However, in practice, sampling $I$ and $J$ independently from the categorical distributions corresponding to the $p$ and $q$ component, respectively, would not result in a large probability for the event $I = J$, even if the corresponding weights are close and using common random numbers. 

Thankfully, it is easy to sample from a maximal coupling of multinomial distributions~\citep[see, e.g.,][Section 4.2]{Thorisson2000Coupling}, and we can therefore write a coupled version of the ensemble rejection sampler. This is summarized in Algorithm~\ref{alg:ensemble-rejection-coupling}.

\begin{algorithm}[!htb]
    \caption{Ensemble rejection coupling of $(p, q)$}
    \label{alg:ensemble-rejection-coupling}
    \DontPrintSemicolon
    \Fn{\ERC{$\hat{\Gamma}$, $p$, $q$, $N$}}{
        \tcp{Operations involving the index $i$ need to be done $N$ times}
        Set $A_X = 0$ and $A_Y = 0$ \tcp{Acceptance flags}
        \While{
            $A_X = 0$ and $A_Y = 0$
        }
        {
            Sample $\hat{X}_i, \hat{Y}_i \sim \hat{\Gamma}$, $U \sim \mathcal{U}(0, 1)$\;
            Set $w_i^X = \frac{p(\hat{X}_i)}{\hat{p}(\hat{X}_i)}$, $w_i^Y =\frac{q(\hat{Y}_i)}{\hat{p}(\hat{Y}_i)}$,  $\hat{Z}_X = \frac{1}{N} \sum_{n=1}^N w_n^X$, and $\hat{Z}_Y = \frac{1}{N} \sum_{n=1}^N w_n^Y$\;
            Set $W_i^X = \frac{w_i^X}{N \hat{Z}_X}$ and  $W_i^Y = \frac{w_i^Y}{N \hat{Z}_Y}$\;
            Sample $(I, J)$ from a maximal coupling of $\mathrm{Cat}((W^X_i)_{i=1}^N)$ and  $\mathrm{Cat}((W^Y_i)_{i=1}^N)$\;
            Set $\bar{Z}_X = \hat{Z}_X + \frac{1}{N}(M(p, \hat{p}) - w_I^X)$ and $\bar{Z}_Y = \hat{Z}_Y + \frac{1}{N}(M(q, \hat{q}) - w_J^Y)$\;
            Set $X_1 = \hat{X}_I$ and $Y_1 = \hat{Y}_J$\;
            \lIf{
                $U < \frac{\hat{Z}_X}{\bar{Z}_X}$
            }{
                set $A_X = 1$
            }
            \lIf{
                $U < \frac{\hat{Z}_Y}{\bar{Z}_Y}$
            }{
                set $A_Y = 1$
            }
        }
        Sample $X_2, Y_2$ from $p \otimes q$\;
        \Ret{$X = A_X \, X_1 + (1 - A_X) \, X_2$, $Y = A_Y \, Y_1 + (1 - A_Y) \, Y_2$}
    }
\end{algorithm}

In the remainder of this section, all the quantities $X$, $Y$, $\mathbb{P}(X=Y)$, and $\mathbb{P}(I=J)$ implicitly depend on the number of proposals $N$ in the sampled ensemble. However, we do not notationally emphasize this dependency for the sake of readibility.

Similarly to Proposition~\ref{prop:rejection-coupling}, if we now define $\tau$ to be the number of times we sample $N$ i.i.d. variables from $\hat{\Gamma}$, we have the following result.
\begin{proposition}
    \label{prop:ensemble-rejection-coupling}
    Suppose that $\hat{\Gamma} \succeq p \otimes q$. Then, for any $N>0$, the random variables $X$ and $Y$ generated by Algorithm~\ref{alg:ensemble-rejection-coupling} are distributed according to $p$ and $q$, respectively. Furthermore, the mean and variance of the random number of steps $\tau$ of Algorithm~\ref{alg:ensemble-rejection-coupling} verify
    \begin{equation}\label{eq:bound_nb_steps}
        \begin{split}
        \mathbb{E}[\tau] &\leq \frac{N + \min\left(M(p, \hat{p}), M(q, \hat{q})\right) - 1}{N}, \\
        \mathbb{V}[\tau] &\leq \left(\frac{N + \min\left(M(p, \hat{p}), M(q, \hat{q})\right) - 1}{N}\right)^2 - \frac{N + \min\left(M(p, \hat{p}), M(q, \hat{q})\right) - 1}{N}.
        \end{split}
    \end{equation}
\end{proposition}

\begin{proof}
    As outlined in the proof of Proposition 1 in \citet{Deligiannidis2020ensemble}, the ensemble rejection sampler is a rejection sampler on an extended space. As a consequence, because marginally both $X$ and $Y$ in Algorithm~\ref{alg:ensemble-rejection-coupling} are sampled from an ensemble rejection sampler, the analysis in the proof of Proposition~\ref{prop:rejection-coupling} still holds. The marginals are therefore preserved, and the bounds \eqref{eq:bound_nb_steps} follow from replacing $p_{\mathrm{RS}}$ by $p_{\mathrm{ERS}}$.
\end{proof}

Moreover Algorithm~\ref{alg:ensemble-rejection-coupling} preserves the diagonal property as well, and, in fact, contrarily to Algorithm~\ref{alg:rejection-coupling}, the coupling success probability of Algorithm~\ref{alg:ensemble-rejection-coupling}  can be quantified, and we can show that it asymptotically recovers the maximal coupling probability. This result is given in the following proposition and subsequent corollary.
\begin{proposition}\label{prop:asymptotics-bound} 
    Let $(\hat{X}, \hat{Y}) \sim \hat{\Gamma}$ be distributed according to the dominating coupling. Suppose that the hypotheses of Proposition~\ref{prop:diagonal-rejection-coupling} are verified, then there exists a deterministic sequence $l_N$, such that for all $N$ we have $0 < l_N \leq \mathbb{P}(X = Y)$ and
    \begin{align}
        l_N &\sim \mathbb{P}(\hat{X} = \hat{Y}) \cdot \frac{\int \min(p(x), q(x))\dd{x}}{1 + u \cdot \frac{\sqrt{2 N \log \log N}}{N}} \cdot \frac{N}{N - 1 + \frac{1}{p_{X}}} \cdot \frac{N}{N - 1 + \frac{1}{p_{Y}}},
    \end{align}
    where $u = \max\left(\mathbb{V}\left[\frac{p(\hat{X})}{\hat{p}(\hat{X})}\right]^{1/2}, \mathbb{V}\left[\frac{q(\hat{Y})}{\hat{q}(\hat{Y})}\right]^{1/2}\right)$, and
    \begin{equation}
        p_X = \mathbb{P}\left( \frac{p(\hat{X})}{M(p, \hat{p}) \hat{p}(\hat{X})}\mid \hat{X} = \hat{Y}\right) > 0, \quad
        p_Y = \mathbb{P}\left(\frac{q(\hat{Y})}{M(q, \hat{q}) \hat{q}(\hat{X})}\mid \hat{X} = \hat{Y}\right) > 0
    \end{equation}
    are the conditional acceptance probabilities of the individual rejection samplers. 
    Additionally, when the distributions $p$ and $q$ do not have atoms, there exists a sequence $u_N$, such that for all $N$, $\mathbb{P}(X = Y) \leq u_N$ and
    \begin{align}
        u_N &\sim \mathbb{P}(\hat{X} = \hat{Y}) \cdot \frac{\int \max(p(x), q(x))\dd{x}}{1 - u \cdot \frac{\sqrt{2 N \log \log N}}{N}}.
    \end{align}
    Finally, when the coupling $\hat{\Gamma}$ is maximal, we have $p_X \geq \frac{\mathbb{E}\left[p(\hat{Y})\right]}{\mathbb{P}(\hat{X} = \hat{Y}) \, M(p, \hat{p})}$, and similarly, $p_Y \geq \frac{\mathbb{E}\left[q(\hat{X})\right]}{\mathbb{P}(\hat{X} = \hat{Y}) \, M(q, \hat{q})}$, with the identity $p_X = \frac{1}{M(p, \hat{p})}$ and $p_Y = \frac{1}{M(q, \hat{q})}$ when we also have $\hat{p} = \hat{q}$.
\end{proposition}
\begin{proof}
    The proof of Proposition~\ref{prop:asymptotics-bound}, relying on the law of iterated logarithm, can be found in {\if1\jasa the supplementary material\else Appendix~\ref{app:proof-asymptotic-bounds}\fi}.
\end{proof}
Proposition~\ref{prop:asymptotics-bound} implies the following corollary.
\begin{corollary}
    Suppose that the hypotheses of Proposition~\ref{prop:diagonal-rejection-coupling} are verified and that the distributions $p$ and $q$ do not have atoms. Using the same notations, we have
    \begin{align}
        \mathbb{P}(X=Y) \underset{N \to \infty}{\longrightarrow} \mathbb{P}(\hat{X} = \hat{Y}) \int \min(p(x), q(x))\dd{x}.
    \end{align}
\end{corollary}

\begin{remark}\label{rem:equal_prop}
    At first sight, the result of Proposition \ref{prop:asymptotics-bound} might seem disappointing. Indeed, we are proving that by throwing an infinite amount of computing power at the coupling algorithm, we can lower bound the coupling success probability by a quantity lower than the maximal one. However, we need to remember that the choice of $\hat{\Gamma}$ is in fact arbitrary, and that we can often choose it in such a way that it increases $\mathbb{P}(\hat{X} = \hat{Y})$. Furthermore, we can also use a common proposal $\gamma$ for $p$ and $q$, in which case using the dominating coupling $(X, Y)$, where we have $X \sim \gamma$ and $Y = X$ results in $\mathbb{P}(\hat{X} = \hat{Y}) = 1$, and therefore the asymptotics of Proposition~\ref{prop:ensemble-rejection-coupling} recover the maximal coupling probability. 
\end{remark}

Proposition~\ref{prop:asymptotics-bound} also provides a clear rule of thumb on how to choose $N$ to achieve a desired coupling probability. For simplicity, suppose that we are using the dominating coupling $\hat{\Gamma}$ corresponding to the duplicated distribution $\gamma$. In this case, in order to obtain a coupling success of $\alpha \cdot \int \min(p(x), q(x))\dd{x}$, $0 < \alpha < 1$, we can simply solve $N$ for 
\begin{equation}
    \alpha = \frac{1}{1 + u \cdot \frac{\sqrt{2 N \log \log N}}{N}} \Leftrightarrow u \cdot \frac{\sqrt{2 N \log \log N}}{N} = \frac{1}{\alpha} - 1,
\end{equation}
which approximately corresponds to $N \approx 2\left(\frac{u}{\frac{1}{\alpha} - 1}\right)^2$.

A natural critique that may be raised when using ensemble methods is its computational cost. In fact, in \citet{Deligiannidis2020ensemble}, the authors argue that using an ensemble proposal to do rejection sampling is less efficient than running a classical rejection sampler for longer. However, we discuss in {\if1\jasa the supplementary material\else Appendix~\ref{app:efficiency}\fi} that this may not be the case when the algorithm is run on parallel hardware, at least for small enough acceptance probabilities. The theoretical analysis done in {\if1\jasa the supplementary material\else Appendix~\ref{app:efficiency}\fi} was also confirmed empirically in Section~\ref{subsec:resampling} (see Figure~\ref{fig:runtime_resampling}).

\section{The multidimensional Gaussian case}
\label{sec:gaussian}
In Sections~\ref{sec:rejection-coupling} and \ref{sec:ensemble-rejection}, we provided a general framework for sampling from diagonal couplings of arbitrary probability distributions by using a dominating coupling $\hat{\Gamma}$ as a proposal. We now focus on the important special case of coupling $d$-dimensional Gaussian distributions with different means and covariances $\mathcal{N}(\mu_p, \Sigma_{p})$ and $\mathcal{N}(\mu_q, \Sigma_{q})$. 

The problem of coupling Gaussian distributions has appeared in \citep{eberle2019couplings, Bou2020coupling} and efficient solutions are known for the case when $\Sigma_p = \Sigma_q$. When it is not the case, one can use Thorisson's algorithm~\citep{Thorisson2000Coupling}, or its finite variance version~\citep{Gerber2020discussion} in order to sample from a (sub-)maximal coupling of these. This last solution, however, presents several inconveniences: first, as already discussed, it has a run time variance that increases as the marginals get closer. While this may be acceptable on distributed systems (such as the one considered in \citet{Jacob2020UnbiasedMCMC}), the context of parallel hardware, which we mainly consider here, requires all parallel executions to finish to proceed. Second, when used with a MCMC algorithm, Thorisson's algorithm does not result in strong contractive properties~\citep{Wang2021maximal}, thereby increasing the total run time of the method, as well as the variance of the estimate. The goal of this section is therefore to provide a fast and efficient way to generalize the reflection approach of \citet{Bou2020coupling,eberle2019couplings} to Gaussians with different covariance so as to reduce the variance of the run time, and preserve the efficient geometric properties that they induce.

Our aim is now to use Algorithm~\ref{alg:rejection-coupling} to sample from a coupling of $\mathcal{N}(\mu_p, \Sigma_{p})$ and $\mathcal{N}(\mu_q, \Sigma_{q})$, with different means and covariances, by using a maximal coupling between Gaussians sharing the same covariance matrix~\citep{Bou2020coupling} as the proposal $\hat{\Gamma}$. For the sake of completeness, we reproduce the method of \citet{Bou2020coupling}, called ``reflection-maximal'' coupling in {\if1\jasa the supplementary material\else Appendix~\ref{app:reflection}\fi}. A natural candidate for a diagonal dominating coupling of $\mathcal{N}(\mu_p, \Sigma_{p})$ and $\mathcal{N}(\mu_q, \Sigma_{q})$ now consists in a reflection-maximal coupling of $\mathcal{N}(\mu_p, \hat{\Sigma})$ and $\mathcal{N}(\mu_q, \hat{\Sigma})$. For this coupling to be dominating, the covariance matrix $\hat{\Sigma}$ needs to verify $\hat{\Sigma}^{-1} \preceq \Sigma_p^{-1}$ and that $\hat{\Sigma}^{-1} \preceq \Sigma_q^{-1}$ in the sense of Loewner ordering. This ensures that we have $\mathcal{N}(x; \mu_p, \Sigma_{p}) \leq \frac{\det(2 \pi \, \hat{\Sigma})^{1/2}}{\det(2 \pi \, \Sigma_{p})^{1/2}} \mathcal{N}(x; \mu_p, \hat{\Sigma})$ and $\mathcal{N}(x; \mu_q, \Sigma_{q}) \leq \frac{\det(2 \pi \, \hat{\Sigma})^{1/2}}{\det(2 \pi \, \Sigma_{q})^{1/2}} \mathcal{N}(x; \mu_q, \hat{\Sigma})$. 

The resulting algorithm now preserves diagonality of the coupling, and we have explicit upper bounds for the mean and variance of the number of steps $\tau$ .

\begin{proposition}\label{prop:mvn-coupling}
    Let $\hat{\Sigma}$ be a covariance matrix with $\hat{\Sigma}^{-1} \preceq \Sigma_p^{-1}$ and $\hat{\Sigma}^{-1} \preceq \Sigma_q^{-1}$, and let $\hat{\Gamma}$ be a reflection-maximal coupling of $\mathcal{N}(\mu_p, \hat{\Sigma})$ and $\mathcal{N}(\mu_q, \hat{\Sigma})$. In this case, we have $M(p,\hat{p}) = \frac{\det(2 \pi \, \hat{\Sigma})^{1/2}}{\det(2 \pi \, \Sigma_{p})^{1/2}}$ and $M(q,\hat{q}) = \frac{\det(2 \pi \, \hat{\Sigma})^{1/2}}{\det(2 \pi \, \Sigma_{q})^{1/2}}$ and Algorithm~\ref{alg:rejection-coupling} for $\hat{\Gamma}$, $\mathcal{N}(\mu_p, \Sigma_{p})$, and $\mathcal{N}(\mu_q, \Sigma_{q})$ returns a diagonal coupling of $\mathcal{N}(\mu_p, \Sigma_{p})$. 
\end{proposition}
\begin{proof}
    This is a direct consequence of Proposition~\ref{prop:rejection-coupling} applied to Gaussian distributions.
\end{proof}

\begin{remark}
    When $\Sigma_p = \Sigma_q$, one can take $\hat{\Sigma} = \Sigma_p = \Sigma_q$, so that the coupling proposal is always accepted and the resulting algorithm recovers the reflection-maximal coupling. 
\end{remark}

A covariance matrix $\hat{\Sigma}$ with the properties $\hat{\Sigma}^{-1} \preceq \Sigma_p^{-1}$ and $\hat{\Sigma}^{-1} \preceq \Sigma_q^{-1}$
always exists, as one can take, for example, $\hat{\Sigma}_{\max} = \max([\sigma(\Sigma_{p}), \sigma(\Sigma_{q})]) \, I_d$, where $\sigma(M)$ is the spectrum of the matrix $M$. Given a chosen matrix $\hat{\Sigma}$, we could, in theory, compute the probability of a successful coupling for Algorithm~\ref{alg:rejection-coupling}, but the expression is computationally intractable. Therefore, we resort to finding lower and upper bounds for the probability $\mathbb{P}(X = Y)$. 
\begin{proposition}
\label{prop:proba}
    Under the hypotheses of Propositions~\ref{prop:mvn-coupling}, we have
    \begin{align} %
        0 &<  \frac{\det(2\pi H)^{1/2}}{\det(2\pi \hat{\Sigma})^{1/2}} \left[\exp(-\beta/2) \mathcal{F}(\alpha, H) + \exp(-\gamma/2) \left[1 - \mathcal{F}(\delta, H)\right]\right] \label{eq:lower-bound} \\
        & \leq \mathbb{P}(X=Y) \leq 2F\left(-\frac{1}{2} \norm{\hat{\Sigma}^{-1/2}(\mu_p - \mu_{q})}_2\right) \label{eq:upper-bound},
    \end{align}
    where 
    \begin{equation}
    \begin{split}
        H &= \left(\Sigma_p^{-1}+ \Sigma_q^{-1} - \hat{\Sigma}^{-1}\right)^{-1},\\
        \alpha &= H\, \left(\Sigma_p^{-1}\, \mu_p + \left(\Sigma_q^{-1} - \hat{\Sigma}^{-1}\right)\, \mu_q\right),\quad \beta= \mu_p^{\top} \Sigma_p^{-1} \mu_p + \mu_q^{\top} \left(\Sigma_q^{-1} - \hat{\Sigma}^{-1}\right) \mu_q - \alpha^{\top} H^{-1} \alpha,\\
        \delta &= H\, \left(\Sigma_q^{-1}\, \mu_q + \left(\Sigma_p^{-1} - \hat{\Sigma}^{-1}\right)\, \mu_p \right),\quad
        \gamma = \mu_q^{\top} \Sigma_q^{-1} \mu_q + \mu_p^{\top} \left(\Sigma_p^{-1} - \hat{\Sigma}^{-1}\right) \mu_p - \delta^{\top} H^{-1} \delta, 
    \end{split}
    \end{equation}
    along with
    \begin{align}
        \mathcal{F}(u, V) \coloneqq F\left(\frac{1}{2}\frac{\mu_p^\top \hat{\Sigma}^{-1} \mu_p - \mu_q^\top \hat{\Sigma}^{-1} \mu_q - 2 u^{\top} \hat{\Sigma}^{-1}(\mu_p-\mu_{q})}{\norm{ \left[ V^{1/2} \right]^\top \hat{\Sigma}^{-1}(\mu_p - \mu_{q})}_2}\right),
    \end{align}
    where $F$ is the cumulative distribution function of a standard normal random variable.
\end{proposition}
\begin{proof}
    The right hand side directly follows from the fact that the rejection coupling is at best as successful as the dominating diagonal coupling used, which in this case is maximal with success probability given by $2F\left(-\frac{1}{2} \norm{\hat{\Sigma}^{-1/2}(\mu_p - \mu_{q})}_2\right)$.
    The proof of the left hand side can be found in {\if1\jasa the supplementary material\else Appendix~\ref{app:proof-lower-bound}\fi}.
\end{proof}

\begin{remark}
The bound in Proposition~\ref{prop:proba} is tight in the sense that the lower and upper bounds coincide when $\Sigma_p = \Sigma_q = \hat{\Sigma}$. In that case, $H=P$, $\alpha=\mu_p$, $\beta=0$, $\delta=\mu_q$, and $\gamma=0$, resulting in 
    \begin{align}
        \mathcal{F}(\alpha, H) = F\left(-\frac{1}{2} \norm{\hat{\Sigma}^{-1/2}(\mu_p - \mu_{q})}_2\right), \quad \mathcal{F}(\beta, H) = 1 - F\left(-\frac{1}{2} \norm{\hat{\Sigma}^{-1/2}(\mu_p - \mu_{q})}_2\right),
    \end{align} 
    which give \eqref{eq:upper-bound} when substituted into \eqref{eq:lower-bound}.
\end{remark}

As a direct corollary, Proposition~\ref{prop:proba} provides a new class of upper bounds for the total variation distance between two Gaussians via \eqref{eq:tv_connection}. These upper bounds are always less than $1$. It is however not clear how they compare to already known bounds (see, e.g., \citet{Devroye2018total} for a recent arXiv report on the subject), in particular in the high dimensional regime, and comparing these is left for future work.

An important aspect of the algorithm is the suitable selection of $\hat{\Sigma}$ so that the coupling probability is as high as possible. Although $\hat{\Sigma}_{\max} = \max([\sigma(\Sigma_{p}), \sigma(\Sigma_{q})]) \, I_d$ is one possible choice for $\hat{\Sigma}$, it might not be a good choice in practice as it may result in a low acceptance probability. Another possible option would be to optimize the lower bound in \eqref{eq:lower-bound} with respect to $\hat{\Sigma}$, which unfortunately leads to a computationally heavy (and likely non-convex) optimisation problem. However, a suitable criterion for maximising the chance of coupling is to simply maximize $\mathbb{P}(A_X = 1) \, \mathbb{P}(A_Y = 1)$. As functions of $\hat{\Sigma}^{-1}$ the logarithms of marginal acceptance probabilities are both of form $\frac{1}{2} \log \det (\hat{\Sigma}^{-1}) + \text{constant}$, and therefore a suitable optimisation problem for $\hat{\Sigma}^{-1}$ is
\begin{equation}
\label{eq:opt_cov}
    \begin{split}
    &\max \log \det(\hat{\Sigma}^{-1}) \\
            & \hat{\Sigma}^{-1} \preceq \Sigma_p^{-1}, \quad
                \hat{\Sigma}^{-1} \preceq \Sigma_q^{-1}, \quad
                \hat{\Sigma}^{-1} \succeq 0.
    \end{split}
\end{equation}
\begin{proposition}
\label{prop:explicit_cov}
    A solution to \eqref{eq:opt_cov} is given by 
    \begin{align}
        \hat{\Sigma}_{opt} = C\, V\, U \, V^\top C^\top
    \end{align}
    where $C = \Sigma_q^{1/2}$, $V\, D \, V^\top = C^\top \, \Sigma_p^{-1} \, C$ is such that $V$ is orthonormal and $D$ is diagonal, and $U$ is diagonal with $U_{ii} = 1 / \min(1, D_{ii})$ for all $i$.
\end{proposition}
\begin{proof}
    The proof of Proposition~\ref{prop:explicit_cov} can be found in {\if1\jasa the supplementary material\else Appendix~\ref{app:optimal_cov}\fi}
\end{proof}

The most expensive operation involved in computing the solution covariance matrix given by Proposition~\ref{prop:explicit_cov} is a spectral decomposition, which typically has the same complexity as computing a Cholesky decomposition. As a consequence, computing $\hat{\Sigma}_{opt}$ only incurs a constant additional cost compared to using $\hat{\Sigma}_{max}$. The benefit of optimising for $\hat{\Sigma}$ as compared to simply using $\hat{\Sigma}_{max}$ is illustrated in {\if1\jasa the supplementary material\else Appendix~\ref{app:gauss}\fi}.

\section{Applications}
\label{sec:applications}
To illustrate the scope and attractiveness of our proposed method, we apply it on a number of examples covering extreme events sampling, particle filtering, Gibbs, and manifold MCMC methods. The code used to generate the examples in this section can be found {\if1\jasa in the supplementary material\else at the following \href{https://github.com/AdrienCorenflos/coupled_rejection_sampling/}{\textbf{address}}\fi}. For consistency, all the examples in this Section were run using a Nvidia GeForce GTX 1080 Ti graphics processing unit with 3584 cores and 11 GB of memory.

\subsection{Coupling tails of Gaussian distributions} \label{sec:gauss_tails}
    The first illustrative example consists in coupled sampling from the tails of Gaussian distributions. This sampling is known to be challenging even without the coupling \citep{Robert1995simulation}. In particular, let $X \sim \mathcal{N}(0, 1)$ and $Y \sim \mathcal{N}(0, 1)$ be Gaussian random variables and let $p_{\mu}(x)$ and $q_{\eta}(y)$ be the densities corresponding to random variables $X_\mu = X \mid X > \mu$ and $Y_\eta = Y \mid Y > \eta$, respectively. When $\mu$ and $\eta$ are small enough, we can sample from $p$ and $q$ by computing the corresponding inverse CDF formulas. However, when $\mu$ and $\eta$ are large ( $\gtrsim 5$) this method becomes unreliable and rejection sampling methods are needed~\citep{Robert1995simulation}. Sampling from a maximal coupling of $p_\mu$ and $q_\eta$ in this case would in fact be even harder than sampling from the marginals independently. Fortunately, we can use our methodology to construct a coupled version of the Gaussian tail rejection sampling algorithm in \citet{Robert1995simulation}, by constructing a maximal coupling of translated exponentials as the dominating coupling. 

    \citet{Robert1995simulation} uses translated exponential distributions $\mathrm{Exp}(x; m,\lambda) = \lambda e^{-\lambda(x - m)} \mathbbm{1}_{x \geq m}$ as the proposal distributions for rejection sampling from the Gaussian tails. Let us now suppose that $\eta > \mu$. As per \citet{Robert1995simulation}, the optimal exponential proposals for $p(x)$ and $q(y)$ are given by $\hat{p}(x) = \mathrm{Exp}(x; \mu,\alpha(\mu))$ and $\hat{q}(x) = \mathrm{Exp}(x; \eta,\alpha(\eta))$, where $\alpha(z) = \frac{z + \sqrt{z^2 + 4}}{2}$ for all $z$ and hence we have 
    \begin{equation}\label{eq:exp_ratio}
        \frac{p(x)}{M(p, \hat{p}) \hat{p}(x)} 
            = e^{-\frac{1}{2}(x - \alpha(\mu))^2}, \quad \frac{q(x)}{M(q, \hat{q}) \hat{q}(x)}
            = e^{-\frac{1}{2}(x - \alpha(\eta))^2}.
    \end{equation}
    Sampling from a maximal coupling of $\hat{p}$ and $\hat{q}$ can then be done explicitly by using the following algorithm (see, e.g, \citep{Thorisson2000Coupling,Lindvall2002lectures}):
\begin{enumerate}
\item Draw $u \sim \mathcal{U}(0,1)$, then
\begin{enumerate}
\item If $u \le \int \min( \hat{p}(x), \hat{q}(x) ) \, dx$: sample $x$ from $c(x) \propto \min( \hat{p}(x), \hat{q}(x) )$ and put $y = x$.
\item Else sample $x$ from $\tilde{p}(x) \propto \hat{p}(x) - \min( \hat{p}(x), \hat{q}(x) )$ and $y$ from $\tilde{q}(y) \propto \hat{q}(y) - \min( \hat{p}(y), \hat{q}(y) )$.
\end{enumerate}
\end{enumerate}
All the probabilities and probability densities above can be evaluated in closed form and the samplings can be done with inverse CDF method. The only difficult parts are the distributions $\tilde{p}(x)$ and $\tilde{q}(y)$ whose inverse CDFs are not available in closed form. However, sampling from these distributions can be implemented with a fast line search method which numerically solves the values of the inverse CDF, see {\if1\jasa the supplementary material\else Appendix~\ref{app:max_coupling_expon}\fi} for more details on the implementation. 
 
    The above method along with the ratios defined in \eqref{eq:exp_ratio} allow us to use Algorithm~\ref{alg:rejection-coupling} or \ref{alg:ensemble-rejection-coupling} to sample from a diagonal coupling of Gaussian tails without ever having to compute the CDF or the inverse CDF, as per the motivation in \citet{Robert1995simulation}. 
    In order to assess the performance of the resulting coupled rejection-sampler, we use Algorithm~\ref{alg:rejection-coupling} to sample $100\,000$ from a coupling of $p_{6}$ and $q_{\eta}$ for different values of $\eta \geq 6$ using the dominating coupling given by the translated exponentials described in \eqref{eq:exp_ratio}. 
    
    In Figure~\ref{fig:gauss_tails_coupling} we report both the maximal coupling probability of $p_{6}$ and $q_{\eta}$, as computed numerically by $\int \min(p_{6}, q_{\eta})$, as well as the empirical coupling probability recovered by our coupled rejection-sampler. In Figure~\ref{fig:gauss_tails_runtime} we also report the resulting run time for different values of the truncation constants. As shown by Figure~\ref{fig:gauss_tails_coupling}, for this specific example, the coupling probability of the samples obtained by Algorithm~\ref{alg:rejection-coupling} is essentially indistinguishable from its theoretical maximum. Furthermore, Figure~\ref{fig:gauss_tails_runtime} shows that the resulting run time of the method is well-behaved, even as $\eta \to 6$, which in the case of using Thorisson's algorithm would have resulted in a run time with infinite variance (see Section~\ref{sec:intro}).
    
    \begin{figure}[!htb]
         \centering
         \begin{subfigure}[b]{0.45\textwidth}
             \centering
             \resizebox{\textwidth}{!}{
             \input{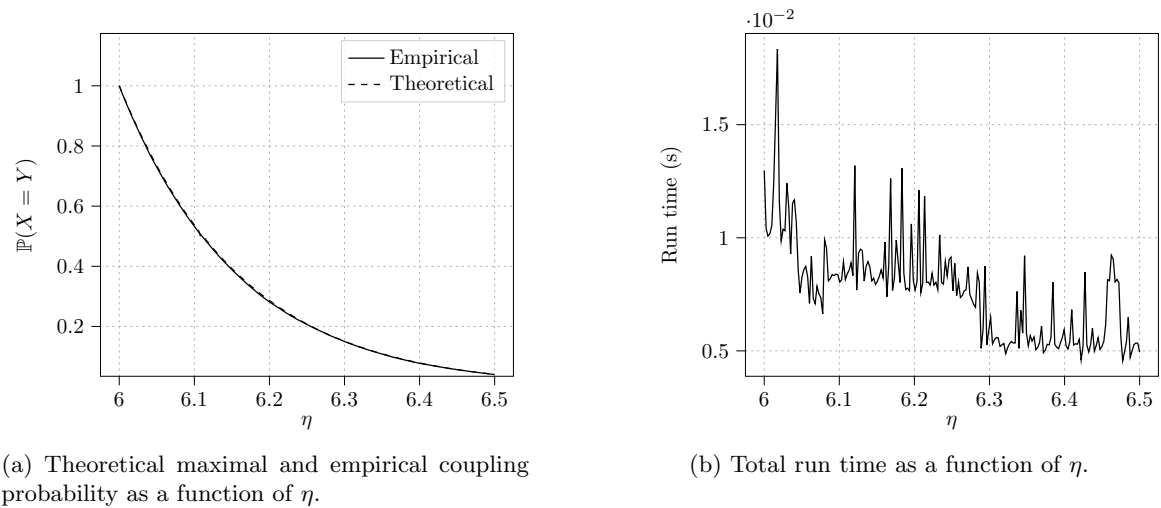}
             }
             \caption{Theoretical maximal and empirical coupling probability as a function of $\eta$.}
             \label{fig:gauss_tails_coupling}
         \end{subfigure}
         \hfill
         \begin{subfigure}[b]{0.45\textwidth}
             \centering
            \resizebox{\textwidth}{!}{
             \begin{tikzpicture}

\begin{axis}[
legend style={fill opacity=0.8, draw opacity=1, text opacity=1, draw=white!80!black},
tick align=outside,
tick pos=left,
x grid style={white!69.0196078431373!black},
xlabel={\(\displaystyle \eta\)},
ylabel={Run time (s)},
ymajorgrids=true,
yminorgrids=true,
xmajorgrids=true,
xminorgrids=true,
xmin=5.97500105, xmax=6.52499995,
xtick style={color=black},
y grid style={white!69.0196078431373!black},
ymin=0.00387736558914185, ymax=0.01903156042099,
ytick style={color=black}
]
\addplot [semithick, black, forget plot]
table {%
6.000001 0.0129694938659668
6.00251355778894 0.010408878326416
6.00502611557789 0.0100724697113037
6.00753867336683 0.0101819038391113
6.01005123115578 0.0105302333831787
6.01256378894472 0.0123376846313477
6.01507634673367 0.0153050422668457
6.01758890452261 0.0183427333831787
6.02010146231156 0.0116064548492432
6.0226140201005 0.00987029075622559
6.02512657788945 0.0103812217712402
6.02763913567839 0.0103130340576172
6.03015169346734 0.012408971786499
6.03266425125628 0.0113537311553955
6.03517680904523 0.00928187370300293
6.03768936683417 0.0115230083465576
6.04020192462312 0.0116724967956543
6.04271448241206 0.010683536529541
6.045227040201 0.00858569145202637
6.04773959798995 0.00756001472473145
6.05025215577889 0.00826382637023926
6.05276471356784 0.00858759880065918
6.05527727135678 0.00873565673828125
6.05778982914573 0.00828671455383301
6.06030238693467 0.00710129737854004
6.06281494472362 0.00918745994567871
6.06532750251256 0.00731253623962402
6.06784006030151 0.00709128379821777
6.07035261809045 0.00785255432128906
6.0728651758794 0.00751519203186035
6.07537773366834 0.00735163688659668
6.07789029145729 0.00663065910339355
6.08040284924623 0.00989723205566406
6.08291540703518 0.00954270362854004
6.08542796482412 0.00809550285339355
6.08794052261307 0.00819087028503418
6.09045308040201 0.00837135314941406
6.09296563819095 0.00832867622375488
6.0954781959799 0.00839114189147949
6.09799075376884 0.00837135314941406
6.10050331155779 0.00804424285888672
6.10301586934673 0.00813913345336914
6.10552842713568 0.00893425941467285
6.10804098492462 0.00815224647521973
6.11055354271357 0.00840449333190918
6.11306610050251 0.00856518745422363
6.11557865829146 0.00888681411743164
6.1180912160804 0.00832939147949219
6.12060377386935 0.0131893157958984
6.12311633165829 0.00768470764160156
6.12562888944724 0.00931715965270996
6.12814144723618 0.00949692726135254
6.13065400502513 0.0094304084777832
6.13316656281407 0.00808906555175781
6.13567912060302 0.00873231887817383
6.13819167839196 0.00894284248352051
6.1407042361809 0.00870156288146973
6.14321679396985 0.00809836387634277
6.14572935175879 0.00821185111999512
6.14824190954774 0.00793719291687012
6.15075446733668 0.00807905197143555
6.15326702512563 0.00829553604125977
6.15577958291457 0.00857973098754883
6.15829214070352 0.00823497772216797
6.16080469849246 0.0098109245300293
6.16331725628141 0.00738954544067383
6.16582981407035 0.00877857208251953
6.1683423718593 0.0126187801361084
6.17085492964824 0.00767302513122559
6.17336748743719 0.00821089744567871
6.17588004522613 0.00989985466003418
6.17839260301508 0.00891852378845215
6.18090516080402 0.00802206993103027
6.18341771859297 0.0130691528320312
6.18593027638191 0.0084376335144043
6.18844283417085 0.00771474838256836
6.1909553919598 0.00777077674865723
6.19346794974874 0.0076594352722168
6.19598050753769 0.0106122493743896
6.19849306532663 0.00813746452331543
6.20100562311558 0.00769138336181641
6.20351818090452 0.0080564022064209
6.20603073869347 0.012101411819458
6.20854329648241 0.0075829029083252
6.21105585427136 0.00797581672668457
6.2135684120603 0.011838436126709
6.21608096984925 0.00802493095397949
6.21859352763819 0.00804424285888672
6.22110608542714 0.00791025161743164
6.22361864321608 0.00842452049255371
6.22613120100502 0.00790691375732422
6.22864375879397 0.00803136825561523
6.23115631658291 0.00770711898803711
6.23366887437186 0.0101275444030762
6.2361814321608 0.00802469253540039
6.23869398994975 0.00795316696166992
6.24120654773869 0.00895428657531738
6.24371910552764 0.00831317901611328
6.24623166331658 0.0090327262878418
6.24874422110553 0.00914263725280762
6.25125677889447 0.00765037536621094
6.25376933668342 0.0088646411895752
6.25628189447236 0.0074460506439209
6.25879445226131 0.00804519653320312
6.26130701005025 0.00735616683959961
6.2638195678392 0.00745058059692383
6.26633212562814 0.00765347480773926
6.26884468341709 0.00769495964050293
6.27135724120603 0.0087130069732666
6.27386979899497 0.00748443603515625
6.27638235678392 0.00727510452270508
6.27889491457286 0.00705957412719727
6.28140747236181 0.0069277286529541
6.28392003015075 0.00844049453735352
6.2864325879397 0.00802016258239746
6.28894514572864 0.0051119327545166
6.29145770351759 0.00594854354858398
6.29397026130653 0.0087437629699707
6.29648281909548 0.00527024269104004
6.29899537688442 0.00588369369506836
6.30150793467337 0.00648283958435059
6.30402049246231 0.00530886650085449
6.30653305025126 0.00548601150512695
6.3090456080402 0.00558948516845703
6.31155816582915 0.00557613372802734
6.31407072361809 0.00519728660583496
6.31658328140704 0.00526666641235352
6.31909583919598 0.00532841682434082
6.32160839698492 0.00487685203552246
6.32412095477387 0.00513577461242676
6.32663351256281 0.00531196594238281
6.32914607035176 0.00541496276855469
6.3316586281407 0.00534963607788086
6.33417118592965 0.00534820556640625
6.33668374371859 0.00762486457824707
6.33919630150754 0.00511336326599121
6.34170885929648 0.00678873062133789
6.34422141708543 0.00578665733337402
6.34673397487437 0.00920677185058594
6.34924653266332 0.00578570365905762
6.35175909045226 0.00522351264953613
6.35427164824121 0.00567126274108887
6.35678420603015 0.00543069839477539
6.3592967638191 0.00561356544494629
6.36180932160804 0.00505828857421875
6.36432187939698 0.00515437126159668
6.36683443718593 0.00538039207458496
6.36934699497487 0.00610089302062988
6.37185955276382 0.00492000579833984
6.37437211055276 0.00502967834472656
6.37688466834171 0.00528883934020996
6.37939722613065 0.00526714324951172
6.3819097839196 0.00559282302856445
6.38442234170854 0.00803446769714355
6.38693489949749 0.0052940845489502
6.38944745728643 0.00516152381896973
6.39196001507538 0.0050969123840332
6.39447257286432 0.00534296035766602
6.39698513065327 0.00558614730834961
6.39949768844221 0.00593018531799316
6.40201024623116 0.0052335262298584
6.4045228040201 0.00508499145507812
6.40703536180905 0.00541067123413086
6.40954791959799 0.00682163238525391
6.41206047738693 0.00525546073913574
6.41457303517588 0.00531578063964844
6.41708559296482 0.00529193878173828
6.41959815075377 0.0055077075958252
6.42211070854271 0.00459694862365723
6.42462326633166 0.00522828102111816
6.4271358241206 0.00847887992858887
6.42964838190955 0.00526857376098633
6.43216093969849 0.00495147705078125
6.43467349748744 0.00527095794677734
6.43718605527638 0.00600028038024902
6.43969861306533 0.00499057769775391
6.44221117085427 0.00524377822875977
6.44472372864322 0.00557899475097656
6.44723628643216 0.00506424903869629
6.44974884422111 0.0051872730255127
6.45226140201005 0.0054166316986084
6.454773959799 0.00616955757141113
6.45728651758794 0.008148193359375
6.45979907537688 0.00809931755065918
6.46231163316583 0.00922870635986328
6.46482419095477 0.00905680656433105
6.46733674874372 0.00804972648620605
6.46984930653266 0.00815725326538086
6.47236186432161 0.00801730155944824
6.47487442211055 0.00565814971923828
6.4773869798995 0.00456619262695312
6.47989953768844 0.00496530532836914
6.48241209547739 0.00539612770080566
6.48492465326633 0.00649690628051758
6.48743721105528 0.00471210479736328
6.48994976884422 0.00500369071960449
6.49246232663317 0.00527667999267578
6.49497488442211 0.00534582138061523
6.49748744221105 0.00533699989318848
6.5 0.00494623184204102
};
\end{axis}

\end{tikzpicture}
             }
             \caption{Total run time as a function of $\eta$. \newline}
             \label{fig:gauss_tails_runtime}
         \end{subfigure}
         \caption{Coupling statistics~\ref{fig:gauss_tails_coupling} and run time~\ref{fig:gauss_tails_runtime} (computed over $100\,000$ samples) of using Algorithm~\ref{alg:rejection-coupling} with shifted exponential proposals to sample a coupling of Gaussian tails $p_6$ and $q_{\eta}$. The average coupling probability observed $\mathbb{P}(X=Y)$ is almost indistinguishable from its theoretical maximal for all values of $\eta$. On the other hand, the total run time for sampling from Algorithm~\ref{alg:rejection-coupling} stays limited as $\eta \to 6$.}
        \label{fig:gauss_tails}
    \end{figure}
    
    It is worth noting that even in the case when we are not worried about the variance of the run time when $\eta \to 6$, using Thorisson's algorithm would have required sampling from the marginals $p_6$ and $q_{\eta}$, which would have required running a rejection sampling algorithm anyway and would have required computing the ratio of the truncated normal PDFs.

\subsection{Coupled parallel resampling}
\label{subsec:resampling}
    Resampling is a key ingredient in sequential Monte Carlo (SMC) methods~\citep[for a recent review, see, e.g.,][]{Chopin2020book} as it allows to avoid path degeneracy by forgetting unsuccessful trajectories. Formally, given an empirical distribution with unnormalized weights $(w^Z_i)_{i=1}^M$ and locations $(Z_i)_{i=1}^M$, resampling methods typically form a random empirical distribution with uniform weights $\frac{1}{M}$ and locations $(Z_{B^Z_i})_{i=1}^N$, where $\mathbb{P}(B^Z_i = j) \propto w^Z_j$ for all $i, j=1, \ldots, M$.
    
    While most operations in SMC are trivially parallelisable, standard techniques for resampling, such as multinomial, systematic, or stratified resampling~\cite{Carpenter1999improved} require communication and prevent full parallelisation of the method. In \citet{Lawrence2016Parallel}, however, the authors show how the resampling operation can be performed by a tweaked parallel rejection sampling algorithm provided that one knows an upper bound to the unnormalized weights $(w^Z_i)_{i=1}^M$.
    
    On the other hand, \citet{Chopin2015particle} discuss how we can construct a maximal coupling of multinomial distributions by first computing the overlap term $\int \min(p, q)$ and then sampling from a mixture representation of the maximal coupling (this coupling, also used in \citet{Jacob2020UnbiasedSmoothing}, is however much older and can found at least in \citet{Lindvall2002lectures,Thorisson2000Coupling}). While sampling from the maximal coupling -- once it has been formed -- can be then done in parallel using \citet{Lawrence2016Parallel}, computing the mass overlap term $\int \min(p, q)$ breaks the parallelisation obtained via \citet{Lawrence2016Parallel} and therefore the algorithm needs to be modified to retain parallelisation.
    
    Our Algorithms~\ref{alg:rejection-coupling} and \ref{alg:ensemble-rejection-coupling} can readily be adapted to provide a parallel coupled version of the rejection sampler in \citet{Lawrence2016Parallel}\footnote{The only subtlety residing in the fact that the initial proposal in \citet{Lawrence2016Parallel} is deterministic.} and we summarize the modification in {\if1\jasa the supplementary material\else Appendix~\ref{app:resampling-algo}\fi}. Moreover, using the ensemble version also helps with reducing both the expectation and variance of the individual rejection sampler run time. Because in the parallel context variable length tasks need to be synchronized~\citep{Lawrence2016Parallel}, using an ensemble rejection method should therefore improve the effective performance of the resulting algorithm. This, however, also comes at the cost of additional computation on $M \ll N$ auxiliary weights, thereby resulting in an increased memory consumption as well as a total parallel computational complexity of $O(\log(M))$ instead of $O(1)$.
    
    In order to illustrate our method on this problem, we consider a similar setting as in \citet{Lawrence2016Parallel}. Let $M=2^{14}$, and $x_i$ and $z_i$ be identically distributed samples from $\mathcal{N}(0, 1)$,  for $i=1, \ldots, M$. For $y=0, 1, 2, 3$, we then take the weight function to be $\omega_y(a) = \mathcal{N}(y - a; 0, 1)$, forming two unnormalized samples $w_{i, y}^X = \omega_y(x_i)$ and $w_{i, y}^Z = \omega_y(z_i)$ which are both upper bounded by $1/\sqrt{2 \pi}$. The experiment consists of applying the coupled rejection resampling algorithm $100$ times for the two fixed sets of normally distributed particles and the different values of $y$ and $N$.
    
    \begin{figure}[!htb]
        \centering
        \resizebox{\textwidth}{!}{
        \begin{tikzpicture}

\definecolor{color0}{rgb}{0.12156862745098,0.466666666666667,0.705882352941177}

\begin{groupplot}[
    group style={group size=2 by 2},
    width=10cm,
    height=7cm,
]
\nextgroupplot[
log basis x={2},
scaled x ticks=manual:{}{\pgfmathparse{#1}},
tick align=outside,
tick pos=left,
title={{\large $y=0$}},
ylabel={\(\displaystyle \mathbb{P}(B_i^X = B_i^Z)\)},
x grid style={white!69.0196078431373!black},
xmin=1, xmax=128,
xmode=log,
xtick style={color=black},
xticklabels={},
y grid style={white!69.0196078431373!black},
ymin=0, ymax=1,
ymajorgrids=true,
yminorgrids=true,
xmajorgrids=true,
xminorgrids=true,
ytick style={color=black}
]
\path [draw=color0, fill=color0, opacity=0.75]
(axis cs:1,0.782426595687866)
--(axis cs:1,0.517170667648315)
--(axis cs:2,0.517170667648315)
--(axis cs:2,0.471414625644684)
--(axis cs:4,0.471414625644684)
--(axis cs:4,0.515281736850739)
--(axis cs:8,0.515281736850739)
--(axis cs:8,0.556296586990356)
--(axis cs:16,0.556296586990356)
--(axis cs:16,0.594582498073578)
--(axis cs:32,0.594582498073578)
--(axis cs:32,0.62507951259613)
--(axis cs:64,0.62507951259613)
--(axis cs:64,0.647225260734558)
--(axis cs:128,0.647225260734558)
--(axis cs:128,0.66004478931427)
--(axis cs:128,0.891605496406555)
--(axis cs:128,0.891605496406555)
--(axis cs:128,0.882220506668091)
--(axis cs:64,0.882220506668091)
--(axis cs:64,0.866004586219788)
--(axis cs:32,0.866004586219788)
--(axis cs:32,0.844777882099152)
--(axis cs:16,0.844777882099152)
--(axis cs:16,0.8136887550354)
--(axis cs:8,0.8136887550354)
--(axis cs:8,0.778778254985809)
--(axis cs:4,0.778778254985809)
--(axis cs:4,0.744224965572357)
--(axis cs:2,0.744224965572357)
--(axis cs:2,0.782426595687866)
--(axis cs:1,0.782426595687866)
--cycle;

\addplot [semithick, color0, const plot mark left]
table {%
1 0.649798631668091
2 0.607819795608521
4 0.647029995918274
8 0.684992671012878
16 0.719680190086365
32 0.745542049407959
64 0.764722883701324
128 0.775825142860413
};
\addplot [thick, black, dashed]
table {%
0.784584097896751 0.783500730991364
163.143760296866 0.783500730991364
};

\nextgroupplot[
log basis x={2},
scaled x ticks=manual:{}{\pgfmathparse{#1}},
scaled y ticks=manual:{}{\pgfmathparse{#1}},
tick align=outside,
tick pos=left,
title={{\large $y=1$}},
x grid style={white!69.0196078431373!black},
xmin=1, xmax=128,
xmode=log,
ymajorgrids=true,
yminorgrids=true,
xmajorgrids=true,
xminorgrids=true,
xtick style={color=black},
xticklabels={},
y grid style={white!69.0196078431373!black},
ymin=0, ymax=1,
ytick style={color=black},
yticklabels={}
]
\path [draw=color0, fill=color0, opacity=0.75]
(axis cs:1,0.625294327735901)
--(axis cs:1,0.347095876932144)
--(axis cs:2,0.347095876932144)
--(axis cs:2,0.309666484594345)
--(axis cs:4,0.309666484594345)
--(axis cs:4,0.348267525434494)
--(axis cs:8,0.348267525434494)
--(axis cs:8,0.391602993011475)
--(axis cs:16,0.391602993011475)
--(axis cs:16,0.428562641143799)
--(axis cs:32,0.428562641143799)
--(axis cs:32,0.462950825691223)
--(axis cs:64,0.462950825691223)
--(axis cs:64,0.488470673561096)
--(axis cs:128,0.488470673561096)
--(axis cs:128,0.505023241043091)
--(axis cs:128,0.77067494392395)
--(axis cs:128,0.77067494392395)
--(axis cs:128,0.757264137268066)
--(axis cs:64,0.757264137268066)
--(axis cs:64,0.734768867492676)
--(axis cs:32,0.734768867492676)
--(axis cs:32,0.705719590187073)
--(axis cs:16,0.705719590187073)
--(axis cs:16,0.668128490447998)
--(axis cs:8,0.668128490447998)
--(axis cs:8,0.627059578895569)
--(axis cs:4,0.627059578895569)
--(axis cs:4,0.584139704704285)
--(axis cs:2,0.584139704704285)
--(axis cs:2,0.625294327735901)
--(axis cs:1,0.625294327735901)
--cycle;

\addplot [semithick, color0, const plot mark left]
table {%
1 0.486195087432861
2 0.446903109550476
4 0.487663567066193
8 0.529865741729736
16 0.567141115665436
32 0.598859846591949
64 0.622867405414581
128 0.637849092483521
};
\addplot [thick, black, dashed]
table {%
0.78458409789675 0.650229811668396
163.143760296865 0.650229811668396
};

\nextgroupplot[
log basis x={2},
tick align=outside,
tick pos=left,
title={{\large $y=2$}},
xlabel={\(\displaystyle N\)},
ylabel={\(\displaystyle \mathbb{P}(B_i^X = B_i^Z)\)},
x grid style={white!69.0196078431373!black},
xmin=1, xmax=128,
xmode=log,
xtick style={color=black},
ymajorgrids=true,
yminorgrids=true,
xmajorgrids=true,
xminorgrids=true,
y grid style={white!69.0196078431373!black},
ymin=0, ymax=1,
ytick style={color=black}
]
\path [draw=color0, fill=color0, opacity=0.75]
(axis cs:1,0.380186587572098)
--(axis cs:1,0.138311952352524)
--(axis cs:2,0.138311952352524)
--(axis cs:2,0.116687960922718)
--(axis cs:4,0.116687960922718)
--(axis cs:4,0.131962776184082)
--(axis cs:8,0.131962776184082)
--(axis cs:8,0.15852589905262)
--(axis cs:16,0.15852589905262)
--(axis cs:16,0.186877444386482)
--(axis cs:32,0.186877444386482)
--(axis cs:32,0.211361780762672)
--(axis cs:64,0.211361780762672)
--(axis cs:64,0.232857227325439)
--(axis cs:128,0.232857227325439)
--(axis cs:128,0.25114893913269)
--(axis cs:128,0.521312117576599)
--(axis cs:128,0.521312117576599)
--(axis cs:128,0.502374768257141)
--(axis cs:64,0.502374768257141)
--(axis cs:64,0.47527402639389)
--(axis cs:32,0.47527402639389)
--(axis cs:32,0.443940460681915)
--(axis cs:16,0.443940460681915)
--(axis cs:16,0.409362316131592)
--(axis cs:8,0.409362316131592)
--(axis cs:8,0.374072372913361)
--(axis cs:4,0.374072372913361)
--(axis cs:4,0.352611392736435)
--(axis cs:2,0.352611392736435)
--(axis cs:2,0.380186587572098)
--(axis cs:1,0.380186587572098)
--cycle;

\addplot [semithick, color0, const plot mark left]
table {%
1 0.259249269962311
2 0.234649673104286
4 0.253017574548721
8 0.283944100141525
16 0.315408945083618
32 0.343317896127701
64 0.36761599779129
128 0.386230528354645
};
\addplot [thick, black, dashed]
table {%
0.784584097896751 0.40905499458313
163.143760296866 0.40905499458313
};

\nextgroupplot[
log basis x={2},
scaled y ticks=manual:{}{\pgfmathparse{#1}},
tick align=outside,
tick pos=left,
title={{\large $y=3$}},
xlabel={\(\displaystyle N\)},
x grid style={white!69.0196078431373!black},
xmin=1, xmax=128,
xmode=log,
ymajorgrids=true,
yminorgrids=true,
xmajorgrids=true,
xminorgrids=true,
xtick style={color=black},
y grid style={white!69.0196078431373!black},
ymin=0, ymax=1,
ytick style={color=black},
yticklabels={}
]
\path [draw=color0, fill=color0, opacity=0.75]
(axis cs:1,0.213148266077042)
--(axis cs:1,0.0315953865647316)
--(axis cs:2,0.0315953865647316)
--(axis cs:2,0.0231214016675949)
--(axis cs:4,0.0231214016675949)
--(axis cs:4,0.0227711796760559)
--(axis cs:8,0.0227711796760559)
--(axis cs:8,0.0295374095439911)
--(axis cs:16,0.0295374095439911)
--(axis cs:16,0.0422850996255875)
--(axis cs:32,0.0422850996255875)
--(axis cs:32,0.0546276494860649)
--(axis cs:64,0.0546276494860649)
--(axis cs:64,0.0676607862114906)
--(axis cs:128,0.0676607862114906)
--(axis cs:128,0.077771469950676)
--(axis cs:128,0.294525921344757)
--(axis cs:128,0.294525921344757)
--(axis cs:128,0.276142925024033)
--(axis cs:64,0.276142925024033)
--(axis cs:64,0.254444628953934)
--(axis cs:32,0.254444628953934)
--(axis cs:32,0.231337949633598)
--(axis cs:16,0.231337949633598)
--(axis cs:16,0.210201352834702)
--(axis cs:8,0.210201352834702)
--(axis cs:8,0.196777164936066)
--(axis cs:4,0.196777164936066)
--(axis cs:4,0.196595400571823)
--(axis cs:2,0.196595400571823)
--(axis cs:2,0.213148266077042)
--(axis cs:1,0.213148266077042)
--cycle;

\addplot [semithick, color0, const plot mark left]
table {%
1 0.122371822595596
2 0.109858401119709
4 0.109774172306061
8 0.119869381189346
16 0.136811524629593
32 0.15453614294529
64 0.171901851892471
128 0.186148688197136
};
\addplot [thick, black, dashed]
table {%
0.78458409789675 0.217771664261818
163.143760296865 0.217771664261818
};
\end{groupplot}

\end{tikzpicture}
        }
        \caption{Coupling probability $\mathbb{P}(B_i^X = B_i^Z)$ for the parallel rejection resampling in Section~\ref{subsec:resampling} as a function of the number of proposals $N$ in Algorithm~\ref{alg:ensemble-rejection-coupling}. The blue full line and shaded area correspond to the average coupling probability of the resampling experiments, as well as the 95\% confidence interval of these. The dashed line corresponds to the theoretical maximal coupling probability.}
        \label{fig:coupling_proba_resampling}
    \end{figure}
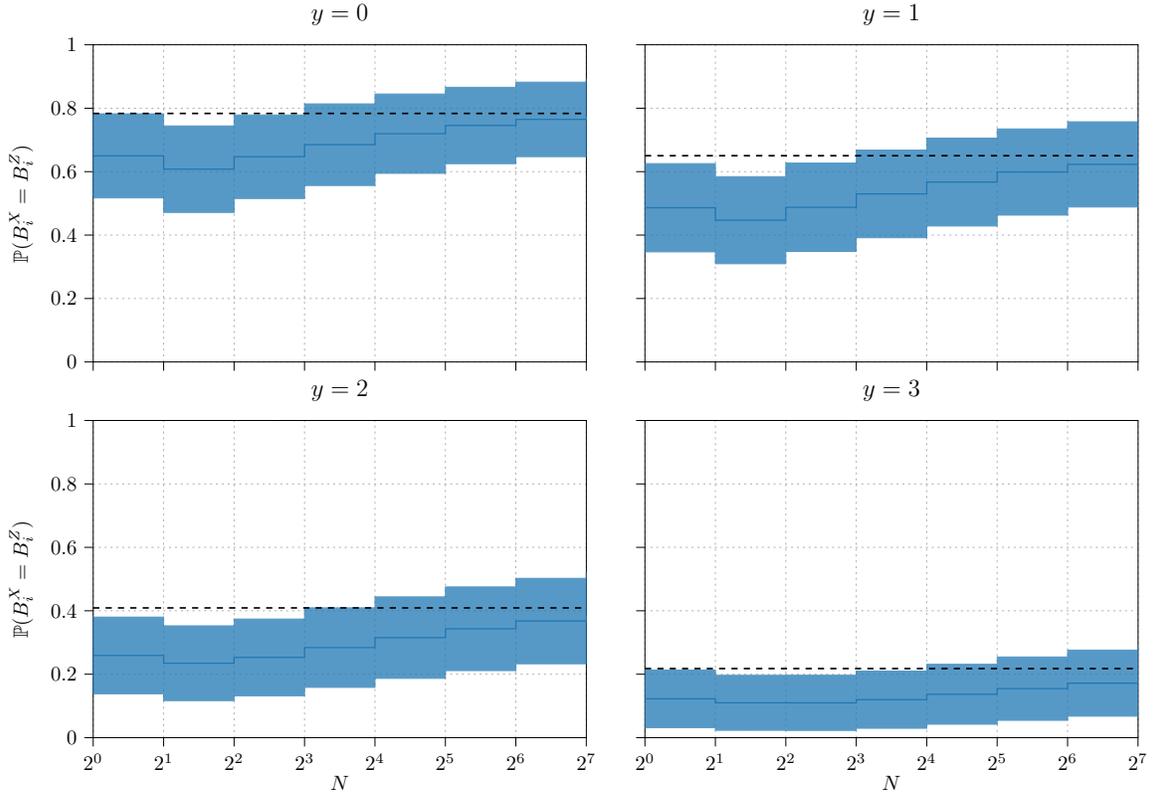
    
    \begin{figure}[!htb]
        \centering
        \resizebox{\textwidth}{!}{
        \begin{tikzpicture}

\definecolor{color0}{rgb}{0.12156862745098,0.466666666666667,0.705882352941177}

\begin{groupplot}[
    group style={group size=2 by 2},
    width=10cm,
    height=7cm,
]
\nextgroupplot[
log basis x={2},
log basis y={10},
scaled x ticks=manual:{}{\pgfmathparse{#1}},
tick align=outside,
tick pos=left,
title={{\large $y=0$}},
x grid style={white!69.0196078431373!black},
xmin=1, xmax=128,
ymajorgrids=true,
yminorgrids=true,
xmajorgrids=true,
xminorgrids=true,
xmode=log,
xtick style={color=black},
xticklabels={},
y grid style={white!69.0196078431373!black},
ymin=9.9817427218364e-05, ymax=0.00108276417488752,
ymode=log,
ytick style={color=black}
]
\path [draw=color0, fill=color0, opacity=0.75]
(axis cs:1,0.000144295292557217)
--(axis cs:1,0.000111241468403023)
--(axis cs:2,0.000111241468403023)
--(axis cs:2,0.000122661411296576)
--(axis cs:4,0.000122661411296576)
--(axis cs:4,0.000111487010144629)
--(axis cs:8,0.000111487010144629)
--(axis cs:8,0.000111825349449646)
--(axis cs:16,0.000111825349449646)
--(axis cs:16,0.000115770519187208)
--(axis cs:32,0.000115770519187208)
--(axis cs:32,0.000132314409711398)
--(axis cs:64,0.000132314409711398)
--(axis cs:64,0.000152684850036167)
--(axis cs:128,0.000152684850036167)
--(axis cs:128,0.00020219900761731)
--(axis cs:256,0.00020219900761731)
--(axis cs:256,0.000295975740300491)
--(axis cs:512,0.000295975740300491)
--(axis cs:512,0.000393445952795446)
--(axis cs:512,0.000633380259387195)
--(axis cs:512,0.000633380259387195)
--(axis cs:512,0.000408564781537279)
--(axis cs:256,0.000408564781537279)
--(axis cs:256,0.000277489045402035)
--(axis cs:128,0.000277489045402035)
--(axis cs:128,0.000198874025954865)
--(axis cs:64,0.000198874025954865)
--(axis cs:64,0.000168523183674552)
--(axis cs:32,0.000168523183674552)
--(axis cs:32,0.000136490489239804)
--(axis cs:16,0.000136490489239804)
--(axis cs:16,0.000130125670693815)
--(axis cs:8,0.000130125670693815)
--(axis cs:8,0.000131702196085826)
--(axis cs:4,0.000131702196085826)
--(axis cs:4,0.000151708634803072)
--(axis cs:2,0.000151708634803072)
--(axis cs:2,0.000144295292557217)
--(axis cs:1,0.000144295292557217)
--cycle;

\addplot [semithick, color0, const plot mark left]
table {%
1 0.000127768376842141
2 0.000137185023049824
4 0.000121594603115227
8 0.000120975506433751
16 0.000126130500575528
32 0.000150418796692975
64 0.000175779437995516
128 0.000239844026509672
256 0.000352270260918885
512 0.000513413106091321
};

\nextgroupplot[
log basis x={2},
log basis y={10},
scaled x ticks=manual:{}{\pgfmathparse{#1}},
scaled y ticks=manual:{}{\pgfmathparse{#1}},
tick align=outside,
tick pos=left,
title={{\large $y=1$}},
x grid style={white!69.0196078431373!black},
xmin=1, xmax=128,
xmode=log,
ymajorgrids=true,
yminorgrids=true,
xmajorgrids=true,
xminorgrids=true,
xtick style={color=black},
xticklabels={},
y grid style={white!69.0196078431373!black},
ymin=9.9817427218364e-05, ymax=0.00108276417488752,
ymode=log,
ytick style={color=black},
yticklabels={}
]
\path [draw=color0, fill=color0, opacity=0.75]
(axis cs:1,0.000139749405207112)
--(axis cs:1,0.00012084480113117)
--(axis cs:2,0.00012084480113117)
--(axis cs:2,0.000141240467200987)
--(axis cs:4,0.000141240467200987)
--(axis cs:4,0.000122172627015971)
--(axis cs:8,0.000122172627015971)
--(axis cs:8,0.000120544304081704)
--(axis cs:16,0.000120544304081704)
--(axis cs:16,0.000123509627883323)
--(axis cs:32,0.000123509627883323)
--(axis cs:32,0.000144959834869951)
--(axis cs:64,0.000144959834869951)
--(axis cs:64,0.000166956349858083)
--(axis cs:128,0.000166956349858083)
--(axis cs:128,0.000206928802072071)
--(axis cs:256,0.000206928802072071)
--(axis cs:256,0.000328861497109756)
--(axis cs:512,0.000328861497109756)
--(axis cs:512,0.000425696140155196)
--(axis cs:512,0.000711430911906064)
--(axis cs:512,0.000711430911906064)
--(axis cs:512,0.000397865689592436)
--(axis cs:256,0.000397865689592436)
--(axis cs:256,0.000267764116870239)
--(axis cs:128,0.000267764116870239)
--(axis cs:128,0.000198684792849235)
--(axis cs:64,0.000198684792849235)
--(axis cs:64,0.000176503090187907)
--(axis cs:32,0.000176503090187907)
--(axis cs:32,0.00014775940508116)
--(axis cs:16,0.00014775940508116)
--(axis cs:16,0.000143756580655463)
--(axis cs:8,0.000143756580655463)
--(axis cs:8,0.000144375269883312)
--(axis cs:4,0.000144375269883312)
--(axis cs:4,0.000166852420079522)
--(axis cs:2,0.000166852420079522)
--(axis cs:2,0.000139749405207112)
--(axis cs:1,0.000139749405207112)
--cycle;

\addplot [semithick, color0, const plot mark left]
table {%
1 0.00013029710680712
2 0.000154046443640254
4 0.000133273948449641
8 0.000132150438730605
16 0.000135634516482241
32 0.000160731462528929
64 0.000182820571353659
128 0.000237346452195197
256 0.000363363593351096
512 0.00056856352603063
};

\nextgroupplot[
log basis x={2},
log basis y={10},
tick align=outside,
tick pos=left,
title={{\large $y=2$}},
xlabel={\(\displaystyle N\)},
x grid style={white!69.0196078431373!black},
xmin=1, xmax=128,
ymajorgrids=true,
yminorgrids=true,
xmajorgrids=true,
xminorgrids=true,
xmode=log,
xtick style={color=black},
y grid style={white!69.0196078431373!black},
ymin=9.9817427218364e-05, ymax=0.00108276417488752,
ymode=log,
ytick style={color=black}
]
\path [draw=color0, fill=color0, opacity=0.75]
(axis cs:1,0.00023681802849751)
--(axis cs:1,0.000191001410712488)
--(axis cs:2,0.000191001410712488)
--(axis cs:2,0.000218106200918555)
--(axis cs:4,0.000218106200918555)
--(axis cs:4,0.000168534927070141)
--(axis cs:8,0.000168534927070141)
--(axis cs:8,0.000163623277330771)
--(axis cs:16,0.000163623277330771)
--(axis cs:16,0.000156846581376158)
--(axis cs:32,0.000156846581376158)
--(axis cs:32,0.000174785935087129)
--(axis cs:64,0.000174785935087129)
--(axis cs:64,0.000198051420738921)
--(axis cs:128,0.000198051420738921)
--(axis cs:128,0.000252929836278781)
--(axis cs:256,0.000252929836278781)
--(axis cs:256,0.000326271634548903)
--(axis cs:512,0.000326271634548903)
--(axis cs:512,0.000579029030632228)
--(axis cs:512,0.000718819617759436)
--(axis cs:512,0.000718819617759436)
--(axis cs:512,0.000459457223769277)
--(axis cs:256,0.000459457223769277)
--(axis cs:256,0.000319429818773642)
--(axis cs:128,0.000319429818773642)
--(axis cs:128,0.000239924731431529)
--(axis cs:64,0.000239924731431529)
--(axis cs:64,0.000221613678149879)
--(axis cs:32,0.000221613678149879)
--(axis cs:32,0.00019525493553374)
--(axis cs:16,0.00019525493553374)
--(axis cs:16,0.000198599183931947)
--(axis cs:8,0.000198599183931947)
--(axis cs:8,0.000233911792747676)
--(axis cs:4,0.000233911792747676)
--(axis cs:4,0.000287614006083459)
--(axis cs:2,0.000287614006083459)
--(axis cs:2,0.00023681802849751)
--(axis cs:1,0.00023681802849751)
--cycle;

\addplot [semithick, color0, const plot mark left]
table {%
1 0.000213909719604999
2 0.000252860103501007
4 0.000201223359908909
8 0.000181111230631359
16 0.000176050758454949
32 0.000198199806618504
64 0.000218988076085225
128 0.000286179827526212
256 0.00039286442915909
512 0.000648924324195832
};

\nextgroupplot[
log basis x={2},
log basis y={10},
scaled y ticks=manual:{}{\pgfmathparse{#1}},
tick align=outside,
tick pos=left,
title={{\large $y=3$}},
xlabel={\(\displaystyle N\)},
x grid style={white!69.0196078431373!black},
xmin=1, xmax=128,
xmode=log,
xtick style={color=black},
y grid style={white!69.0196078431373!black},
ymin=9.9817427218364e-05, ymax=0.00108276417488752,
ymode=log,
ymajorgrids=true,
yminorgrids=true,
xmajorgrids=true,
xminorgrids=true,
ytick style={color=black},
yticklabels={}
]
\path [draw=color0, fill=color0, opacity=0.75]
(axis cs:1,0.000663938233628869)
--(axis cs:1,0.000506061129271984)
--(axis cs:2,0.000506061129271984)
--(axis cs:2,0.000577710918150842)
--(axis cs:4,0.000577710918150842)
--(axis cs:4,0.000413532252423465)
--(axis cs:8,0.000413532252423465)
--(axis cs:8,0.000356609962182119)
--(axis cs:16,0.000356609962182119)
--(axis cs:16,0.000300720130326226)
--(axis cs:32,0.000300720130326226)
--(axis cs:32,0.000306800531689078)
--(axis cs:64,0.000306800531689078)
--(axis cs:64,0.000302208791254088)
--(axis cs:128,0.000302208791254088)
--(axis cs:128,0.000364859966794029)
--(axis cs:256,0.000364859966794029)
--(axis cs:256,0.000505393662024289)
--(axis cs:512,0.000505393662024289)
--(axis cs:512,0.000722653465345502)
--(axis cs:512,0.000971568748354912)
--(axis cs:512,0.000971568748354912)
--(axis cs:512,0.00064082135213539)
--(axis cs:256,0.00064082135213539)
--(axis cs:256,0.000478948903037235)
--(axis cs:128,0.000478948903037235)
--(axis cs:128,0.000398017262341455)
--(axis cs:64,0.000398017262341455)
--(axis cs:64,0.000396533694583923)
--(axis cs:32,0.000396533694583923)
--(axis cs:32,0.000391799811040983)
--(axis cs:16,0.000391799811040983)
--(axis cs:16,0.000473028951091692)
--(axis cs:8,0.000473028951091692)
--(axis cs:8,0.000567047507502139)
--(axis cs:4,0.000567047507502139)
--(axis cs:4,0.000801839283667505)
--(axis cs:2,0.000801839283667505)
--(axis cs:2,0.000663938233628869)
--(axis cs:1,0.000663938233628869)
--cycle;

\addplot [semithick, color0, const plot mark left]
table {%
1 0.000584999681450427
2 0.000689775100909173
4 0.000490289879962802
8 0.000414819456636906
16 0.000346259970683604
32 0.0003516671131365
64 0.000350113026797771
128 0.000421904434915632
256 0.00057310750707984
512 0.000847111106850207
};
\end{groupplot}

\end{tikzpicture}
        }
        \caption{Run time (in seconds) for the parallel rejection resampling in Section~\ref{subsec:resampling} as a function of the number of proposals $N$ in Algorithm~\ref{alg:ensemble-rejection-coupling}. The blue full line and shaded area correspond to the average run time of the resampling experiments, as well as its 95\% confidence interval.}
        \label{fig:runtime_resampling}
    \end{figure}
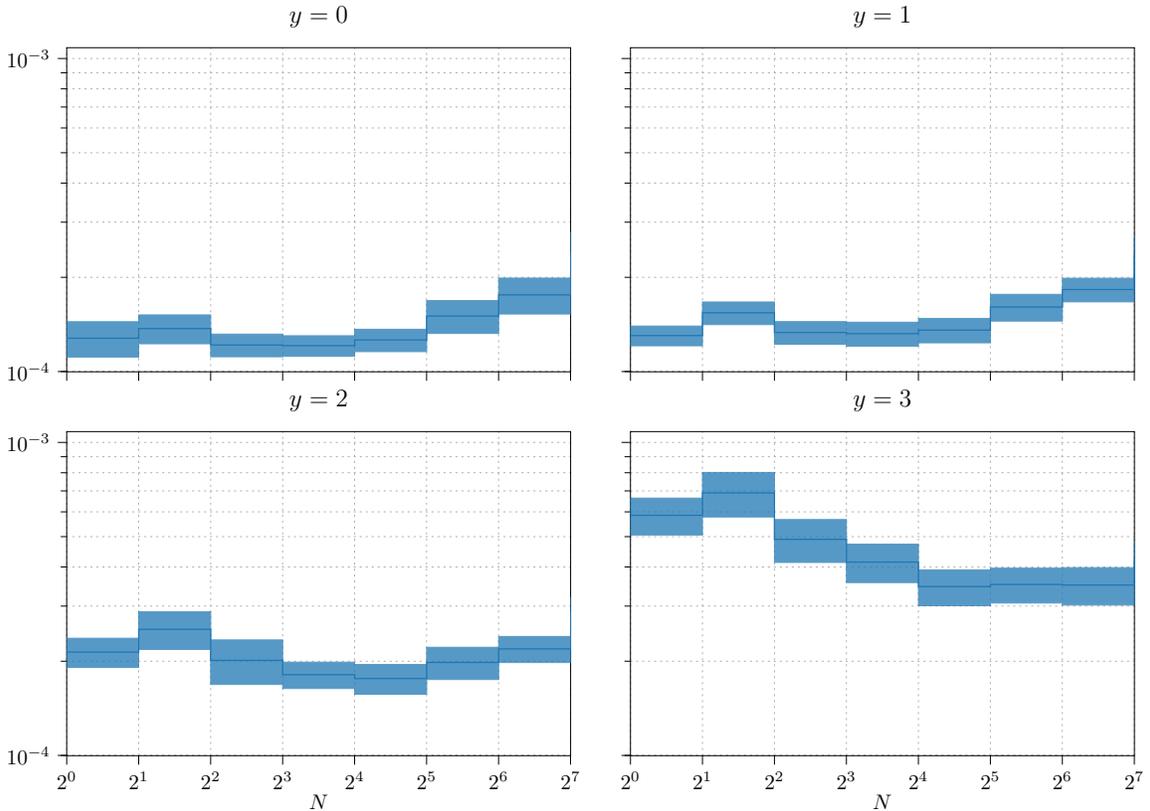
    
   In Figure~\ref{fig:coupling_proba_resampling} and Figure~\ref{fig:runtime_resampling} we respectively report the empirical coupling success $\mathbb{P}(B_i^X = B_i^Z)$ and the run time of running {\if1\jasa the coupled rejection-based resampling algorithm\else Appendix~\ref{alg:coupled-rejection-resampling}\fi} as $y$ changes (and therefore the variance of the weights increases) for different ensemble proposal sizes $N$. In Figure~\ref{fig:coupling_proba_resampling}, the empirical coupling probability is compared with the theoretical maximal coupling probability as given by $\sum_{i}\min(W_{i, y}^X, W_{i, y}^Z)$, where $W_{i, y}^X = \frac{w_{i, y}^X}{\sum_{j=1}^M w_{j, y}^X}$ and $W_{i, y}^Z \frac{w_{i, y}^Z}{\sum_{j=1}^M w_{j, y}^Z}$ are the self-normalized versions of the weights $w_{i, y}^{X}$ and $w_{i, y}^{Z}$. The experiment was run $100$ times.
    
   As can be seen in Figure~\ref{fig:coupling_proba_resampling}, the empirical coupling probability seems to converge to the theoretical one for the large number of proposals $N$ regime. This behavior is in agreement with the asymptotics of Proposition~\ref{prop:asymptotics-bound} as the dominating coupling is here given by the trivial coupling $\hat{X} = \hat{Y} \sim \mathcal{U}(\{1, \ldots, N\})$, so that Remark~\ref{rem:equal_prop} applies. As shown by the average run time reported in Figure~\ref{fig:runtime_resampling} increasing the ensemble size in Algorithm~\ref{alg:ensemble-rejection-coupling}, at the same time, lowers the total run time of the resampling routine for higher weights variance, corresponding to $y=2$ and $y=3$. For a discussion on the computational benefits of using ensemble proposals in the GPU context we refer the reader to {\if1\jasa the supplementary material\else Appendix~\ref{app:efficiency}\fi}.
   
   The run time results reported in Figure~\ref{fig:runtime_resampling} suggest that using an ensemble rejection method in the parallel context may lead to broader applicability of parallel rejection resampling than the recommendation in Section 3.4 of \citet{Lawrence2016Parallel}, in particular in the high weight variance regime. More precisely, when the resampling routine is used as part of a particle filtering loop, it is highly likely that the number of threads necessary to compute the weights is fairly large. This means that, once these have been computed, the number of threads available at resampling time is typically higher than the total number of particles. Allowing to improve the run time by leveraging the idle compute capacity therefore seems valuable in most cases, and particularly so in the coupling context, due to its additional benefit of increasing the coupling probability.

\subsection{Coupled Gibbs sampling}
\label{subsec:Gibbs}
    Let us now consider the following density $p(x, y) \propto \exp\left(-\frac{1}{2}\left(x^{\top}x \cdot y^\top y + x^{\top}x + y^\top y\right)\right)$, $x, y \in \mathbb{R}^d$. 
    We can express the conditional densities $p_{x\mid y}$ and $p_{y \mid x}$ in the following way:
    \begin{align}
        p_{x \mid y}(x \mid y) &= \mathcal{N}(x; 0, P_y), \quad P_y = \frac{1}{1 + y^{\top} y} I_d, \\
        p_{y \mid x}(y \mid x) &= \mathcal{N}(y; 0, P_x), \quad P_x = \frac{1}{1 + x^{\top} x} I_d.
    \end{align}
    This decomposition ensures that we can sample from $p(x, y)$ using Gibbs sampling~\citep[see, e.g.,][Ch. 4]{Robert2004} by sampling from $p_{x \mid y}$ and $p_{y \mid x}$ alternatively. As described in \citet{Jacob2020UnbiasedMCMC}, to couple two resulting Gibbs chains, we only need to couple their conditional transition densities. That is, if we have two states $(x, y)$ and $(x', y')$, we need the ability to sample from a coupling of $p_{x \mid y}$ and $p_{x \mid y'}$ and from a coupling of $p_{y \mid x}$ and $p_{y \mid x'}$. Because the densities have different covariance matrices for different states, we cannot use the reflection-maximal coupling of~\cite{Bou2020coupling} described {\if1\jasa the supplementary material\else Appendix~\ref{app:reflection}\fi}, but,  instead, we need to resort to either Thorisson's algorithm~\cite{Thorisson2000Coupling} in its modified form~\cite{Gerber2020discussion} (so as to have finite run time variance) or to our proposed method. For the sake of completeness, we reproduce the modified version of Thorisson's algorithm in {\if1\jasa the supplementary material\else Appendix~\ref{app:thorisson}\fi}.
    
    \begin{figure}[!htb]
        \centering
        \resizebox{\textwidth}{!}{
        \input{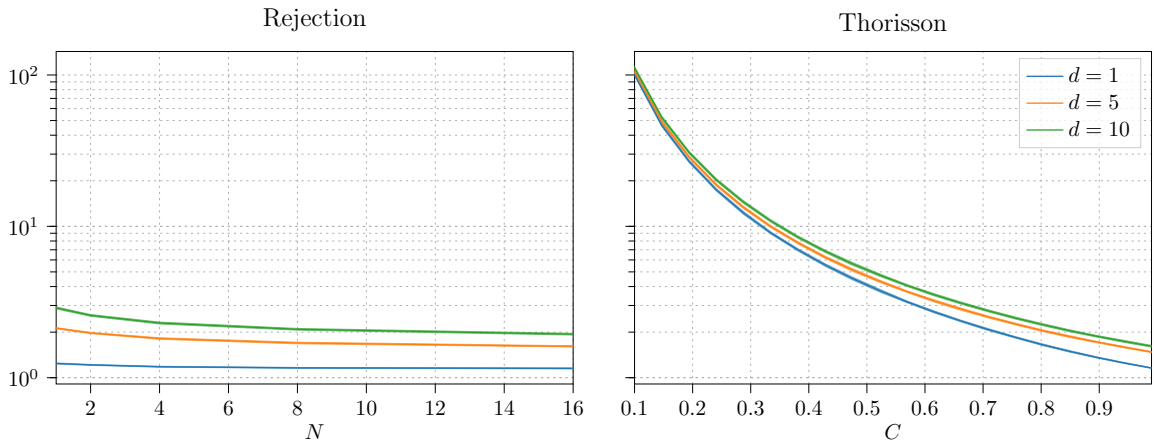}
        }
        \caption{Average coupling time (computed over $50$ experiments) of the coupled Gibbs chains for our rejection method and Thorisson's modified algorithm. The ``Thorisson'' x-axis represents the suboptimality parameter $C$ in the method (see {\if1\jasa the supplementary material\else Appendix~\ref{app:thorisson}\fi}).}
        \label{fig:gibbs_meeting_time}
    \end{figure}
    
    In Figure~\ref{fig:gibbs_meeting_time}, we report the average meeting time~\cite{Jacob2020UnbiasedMCMC} of $10\,000$ coupled chains and, in Figure~\ref{fig:gibbs_runtime}, the average run time of an experiment, as well of the $95$\% confidence of these (we repeated the experiment $50$ times to measure the uncertainty in these quantities). Our results were generated by using the dominating covariance matrix $Q_{max}$, which in this case actually corresponds to the optimal matrix $Q_{opt}$. 
    
    \begin{figure}[!htb]
        \centering
        \resizebox{\textwidth}{!}{
        \input{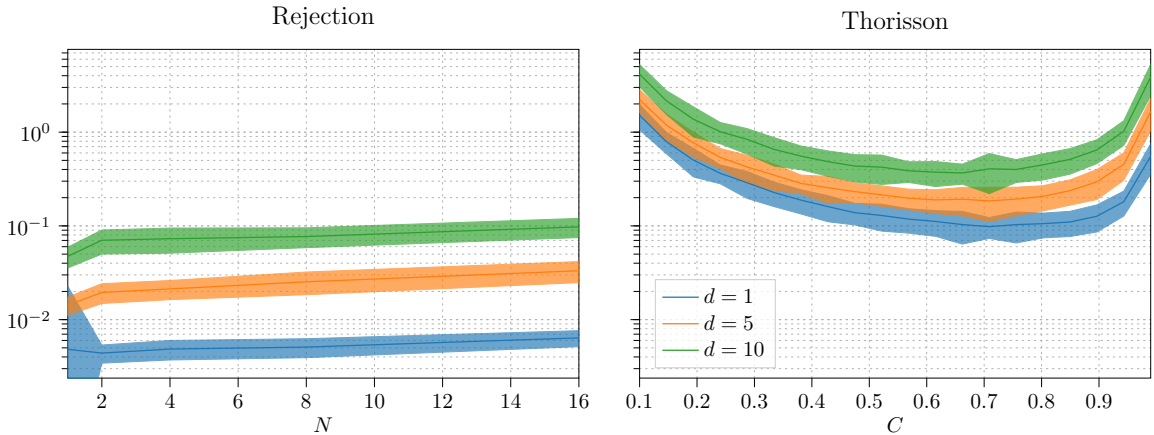}
        }
        \caption{Average time in seconds (computed over $50$ experiments) taken for the coupled Gibbs chains to meet for our rejection method and Thorisson's modified algorithm. The ``Thorisson'' x-axis represents the suboptimality parameter $C$ in the method (see {\if1\jasa the supplementary material\else Appendix~\ref{app:thorisson}\fi}).}
        \label{fig:gibbs_runtime}
    \end{figure}
    
    We can see that our proposed method outperforms the modified version of Thorisson's algorithm in run time (Figure~\ref{fig:gibbs_runtime}) for all sizes of ensembles and almost all values of $C$, on the other hand, Thorisson's algorithm results in slightly faster coupling times for $C \approx 1$, at the cost of a longer running algorithm than our method (due to the increase in variance in the run time). We believe, similarly to the analysis in \citet{Wang2021maximal} that the mostly faster contraction of the Gibbs chains when using our method is explained by the fact that our coupling preserves the reflection properties the reflection-maximal coupling.
        
\subsection{Coupled manifold Markov chain Monte Carlo}
    \label{subsec:mmala}
    Adaptation for the proposal distribution~\citep[for a review see, e.g.,][]{andrieu2008tutorial} typically improves the convergence of MCMC algorithm by adapting the sampling mechanism to the target at hand, either globally, or locally. Coupling two adaptive or pre-conditioned algorithms however requires to be able to couple MCMC methods with different proposal distributions. In the literature, only two methods have been proposed to do this in the general case: coupling proposals using Thorisson's algorithm, or sampling from a maximal coupling of Metropolis--Hastings kernels~\citep{Wang2021maximal}.  
    
    We now consider the simplified manifold Metropolis-adjusted Langevin algorithm (mMALA) applied to the logistic regression example as described in \citet{Girolami2011mMCMC}. Specifically, we consider the ``heart disease'' binary regression dataset~\citep{detrano1989international}, and the same setup as \citet{Girolami2011mMCMC} apart from the starting point of the Markov chains, which we take to be randomly distributed according to $\mathcal{N}(0, 1/4^2)$. Because Riemannian manifold MCMC methods rely on sampling Gaussian distributions, the parameters of which depend on the Fisher information matrix computed at the current state of the chain, applying the reflection-maximal coupling method directly is impossible. This means that we can only hope to couple the two chains by using (i) Thorisson's (modified) algorithm to couple proposals, (ii) the ``full-kernel coupling'' with independent or reflected marginals in \citet[][Algorithms 3 and 4]{Wang2021maximal}, or by coupling the proposals using the coupled rejection sampler introduced in this article~\footnote{It is also worth noting that the empirical evaluation of the different methods proposed in \citet{Wang2021maximal} suffers from an implementation bug, and that, contrarily to what was reported in the article, Algorithm 4 results in faster coupling times than Algorithm 5 in high dimensional problems. This further justifies comparing with Algorithm 4 only.}. In accordance to \citet{Wang2021maximal} and the rest of our experiments, using Thorisson's algorithm underperformed the other alternatives, so we do not report the results below.
    
    We run two mMALA chains with coupled proposals, using either the ``full-kernel coupling with reflection residuals'' Metropolis-Hastings method of \citet[][Algorithm 4]{Wang2021maximal}, or the ensemble coupled rejection sampler with reflection-maximal proposals using the optimal covariance matrix described in Section~\ref{sec:gaussian}. 
    
    \begin{table}[!htb]
        \centering
        \begin{tabular}{llrrrrr}
\toprule
{} & $N$ &  mean &  std &  5\% &  50\% &  95\% \\
\midrule
This paper & 1  &  11.0 & 15.4 & 1 &  5 & 41 \\
     & 4  &   7.7 & 10.8 & 1 &  4 & 28 \\
     & 16 &   6.3 &  8.4 & 1 &  4 & 21 \\
     & 64 &   5.8 &  7.7 & 1 &  3 & 19 \\
\midrule
\citet{Wang2021maximal} &    &   7.5 &  8.5 & 1 &  5 & 24 \\
\bottomrule
\end{tabular}

        \caption{Coupling time statistics of the coupled ensemble rejection sampler on the simplified manifold MALA algorithm of \citep{Girolami2011mMCMC} for different values of $N$ and Algorithm 4 in \citet{Wang2021maximal}. The percentages represent empirical quantiles of the distribution.}
        \label{tab:mmala_coupling}
    \end{table}
        
    \begin{table}[!htb]
        \centering
        \begin{tabular}{llrrrrr}
\toprule
{} & $N$ &    mean &     std &      5\% &     50\% &     95\% \\
\midrule
This paper &1  & 9.3e-03 & 1.3e-02 & 1.0e-03 & 4.4e-03 & 3.5e-02 \\
{} &4  & 7.3e-03 & 1.0e-02 & 1.1e-03 & 3.8e-03 & 2.6e-02 \\
{} &16 & 6.1e-03 & 8.0e-03 & 1.1e-03 & 3.7e-03 & 2.0e-02 \\
{} &64 & 5.8e-03 & 7.4e-03 & 1.1e-03 & 3.5e-03 & 1.9e-02 \\
\midrule
\citet{Wang2021maximal} &   & 1.5e-02 & 1.9e-02 & 6.8e-04 & 8.7e-03 & 5.2e-02 \\
\bottomrule
\end{tabular}

        \caption{Run time statistics of the coupled ensemble rejection sampler on the simplified manifold MALA algorithm of \citep{Girolami2011mMCMC} for different values of $N$ and Algorithm 4 in \citet{Wang2021maximal}. The percentages represent empirical quantiles of the distribution.}
        \label{tab:mmala_runtime}
    \end{table}
    
    In Table~\ref{tab:mmala_coupling} and Table~\ref{tab:mmala_runtime}, we report the resulting meeting times and run times, respectively, over $100\,000$ experiments for the ensemble rejection sampler proposal coupling with different $N$'s as well as the ``full-kernel coupling with reflection residuals''~\citep[][Algorithm 4]{Wang2021maximal}. On average, the coupled rejection sampler method runs faster and results in smaller coupling times Algorithm 4 in \citet{Wang2021maximal} method as soon as $N>4$. As we increase $N$, it also starts exhibiting smaller variance for both quantities, highlighting its attractiveness for implementing coupled chains in parallel environments. This improvement is the same phenomenon as the one discussed in~\citet{Wang2021maximal}: using high quality proposals (when possible) may sometimes beat the full-kernel coupling, even when this one has reflected residuals.

\section{Conclusion and discussion}
\label{sec:discussion}
In this article, we have introduced a novel coupled-rejection algorithm and an ensemble version of it for sampling from a target coupling using an acceptance-rejection scheme on samples from a dominating coupling. We have derived upper and lower bounds for their coupling success probabilities as well as the expectation and variance of their number of steps. In particular, our method provably preserves positive coupling probabilities of dominating coupling proposals and, in the ensemble case, asymptotically recovers the maximal coupling. 

An important special case of our method consists in the multivariate Gaussian case where the algorithm allows to sample from coupled Gaussian distributions with different means and covariances with a low (finite) run time variance. This also allowed us to derive a family of upper bounds to the total variation distance between Gaussians. We believe that these bounds can be further improved by a more careful study of the quantities involved, possibly leading to even more useful bounds. 

The proposed method is, however, not limited to the Gaussian case, and, in fact, if an extension of the reflection-maximal coupling of~\citet{Bou2020coupling} to Gaussians with different means and covariances was derived, it could be used as a proposal in our algorithm. This would in turn simply expand and improve the applications of our proposed method for sampling diagonal couplings of sub-exponential distributions with high success probability.

We have also demonstrated the empirical usefulness of our method in coupling rare events, parallel resampling in particle filtering, Gibbs samplers, as well as Riemannian manifold MCMC methods. Possible future research include applications such as fully parallel unbiased particle smoothing~\citep{Jacob2020UnbiasedSmoothing}, and possible extensions such as coupled empirical rejection sampling or coupling of state-space models for which the ensemble rejection sampling method was originally introduced~\citep{Deligiannidis2020ensemble}. An interesting conclusion of our experimental evaluation, analysed in {\if1\jasa the supplementary material\else Appendix~\ref{app:efficiency}\fi}, is that using ensemble proposals not only increases the coupling probability, as predicted by the theory, but also improves upon the total run time of the method when using parallel hardware such as a GPU. This is the consequence of two compounding factor: first, a computational improvement coming from reducing the variance in the run time of the rejection sampler (Section~\ref{subsec:resampling}, {\if1\jasa the supplementary material\else Appendix~\ref{app:efficiency}\fi}), second, a decrease in the expected meeting time, reducing the total number of operations run before two chains meet (Sections~\ref{subsec:Gibbs} and \ref{subsec:mmala}).

Another direction of work is concerned with the special case of Gaussian distributions. In particular, finding computationally advantageous alternatives to $Q_{opt}$, or even computing a covariance matrix maximising the actual coupling probability instead of the surrogate objective introduced in Section~\ref{sec:gaussian} would result in more successful and computationally feasible coupling. Another question consists in computing tighter bounds to the coupling probability of our algorithm, automatically resulting in better upper bounds for the TV norm between Gaussians. 

It is worth noting that while the focus of this article has been with implementing couplings distributions for use in standard MCMC algorithms, such as Metropolis--Hastings, or Gibbs sampling, \citet{Wang2021maximal} introduced direct maximal couplings of Metropolis--Hastings type kernels. These have however been shown to be suboptimal in higher dimensions, probably due to the lack of contractive behavior of the resulting chains. This behavior was confirmed in Section~\ref{subsec:mmala}, where the ``full-kernel coupling with reflection residuals'' of \citet{Wang2021maximal} under-performed the simple Metropolis--Hastings using our coupled rejection sampling method, thereby hinting to the same conclusion that using good proposals is more important than optimality for coupling MCMC chains. To improve upon these, it could be possible to combine both approaches so as to increase the performance of maximal (or otherwise) couplings of MH kernels, for instance by coupling multiple-try MCMC methods~\citep{Jun2000multi}.

Finally, our method can be used to propagate arbitrary couplings from a dominating pair of marginals to the target ones. For example, suboptimal Wasserstein-like couplings~\cite[for a review see, e.g.,][]{villani2009optimal} could be sampled efficiently from a dominating optimal transport coupling of Gaussian distributions. Similarly, couplings other than the maximal coupling can be used for sampling the indices in Algorithm~\ref{alg:ensemble-rejection-coupling}. For example using optimal transport coupling of the indices based on a well-chosen metric may lead to faster converging algorithms at the cost of submaximality. Studying the theoretical and empirical properties of these different extensions is left for future works.

\section*{Individual contributions}
\label{sec:contribution}
The original idea, implementation, and redaction of this article are all due to Adrien Corenflos. Simo S\"arkk\"a contributed Proposition~\ref{prop:explicit_cov}, and helped with reviewing and finalising the manuscript. Adrien {Corenflos} and Simo S\"arkk\"a jointly developed the example of Section~\ref{sec:gauss_tails}. 

\section*{Acknowledgements}
Adrien Corenflos and Simo S\"arkk\"a gratefully acknowledge the support of Academy of Finland (project 321900). The authors would like to thank Pierre E. Jacob for his valuable comments on the original version of this manuscript. The first author would also like to warmly thank him for making his lecture notes on coupling methods freely available on the internet. This work would most likely never have been written otherwise.

\bibliography{bib}

\newpage
\appendix
\section{Proof of Proposition~\ref{prop:rejection-coupling}}
\label{app:proof_prop_rej}
Let us first prove that the algorithm samples from the correct marginal distributions. In the case when $A_X = 0$, we simply sample from $p$. Otherwise, when $A_X=1$, sampling of $X$ again corresponds to sampling from $p$, but now with rejection sampling with proposal $\hat{p}$ and ratio bound $M(p,\hat{p})$. Therefore the marginal of $X$ is $p$. With a similar reasoning we can deduce that the marginal of $Y$ is $q$.

Now, following classical rejection sampling theory~\citep[see, e.g.,][Ch. 1]{Pages2018}, we know that Algorithm~\ref{alg:rejection-coupling} termination time $\tau$ follows a geometric law with acceptance probability $\frac{1}{\mathbb{E}[\tau]}$.
In order for the algorithm to finish, we need either $A_X = 1$ or $A_Y=1$. This implies that $1 \leq \tau \leq \min(\tau_X, \tau_Y)$ where $\tau_X$ and $\tau_Y$ are the time until acceptance of a sample for $X$ and $Y$, respectively. Similarly as for $\tau$, we know that $\tau_X$ and $\tau_Y$ are marginally distributed according to Geometric laws with expectation $\mathbb{E}\left[\tau_X\right] = M(p, \hat{p})$ and $\mathbb{E}\left[\tau_Y\right] = M(q, \hat{q})$. 
Now, due to the fact that $\tau$ follows a geometric disribution, we know that $\mathbb{V}[\tau] = \mathbb{E}[\tau]^2 - \mathbb{E}[\tau] = \mathbb{E}[\tau](\mathbb{E}[\tau] - 1)$, which yields the result.
Finally, when $p$ and $q$ do not have atoms, then $X = Y$ only happens when $X_1 = Y_1$ and $A_X = 1, A_Y = 1$. This implies that we have $\mathbb{P}(X = Y) \le \mathbb{P}(X_1 = Y_1)$. 

\section{Proof of Proposition~\ref{prop:asymptotics-bound}}
\label{app:proof-asymptotic-bounds}
Similarly as for the proof of Proposition~\ref{prop:diagonal-rejection-coupling}, using the notations of Algorithm~\ref{alg:ensemble-rejection-coupling}, we have
\begin{equation}
\begin{split}
    \mathbb{P}(X = Y) 
        &\geq   \mathbb{P}(\hat{X}_I = \hat{Y}_J \,\&\, A_X = A_Y = 1 \,\&\, I = J) \\ 
        &=      \mathbb{P}(A_X = A_Y = 1 \mid I = J, \hat{X}_I = \hat{Y}_J) \, 
                \mathbb{P}(\hat{X}_I = \hat{Y}_J \mid I = J) \,
                \mathbb{P}(I = J)
\end{split}
\end{equation}
as well as, when the distributions $p$ and $q$ do not have atoms,  
\begin{equation}
\begin{split}
    \mathbb{P}(X = Y) 
        &\leq   \mathbb{P}(\hat{X}_I = \hat{Y}_J \,\&\, I = J) \\ 
        &=      \mathbb{P}(\hat{X}_I = \hat{Y}_J \mid I = J) \,
                \mathbb{P}(I = J).
\end{split}
\end{equation}
Following \citet{Deligiannidis2020ensemble}, we can lower bound $\mathbb{P}(A_X = A_Y = 1 \mid I = J, \hat{X}_I = \hat{Y}_J)$ by 
\begin{equation}
\begin{split}
    &\mathbb{P}(A_X = A_Y = 1 \mid I = J, \hat{X}_I = \hat{Y}_J) \\
    &\geq \mathbb{P}(A_X = 1 \mid I = J, \hat{X}_I = \hat{Y}_J) \cdot \mathbb{P}(A_Y = 1 \mid I = J, \hat{X}_I = \hat{Y}_J) \\
    &\geq \frac{N}{N - 1 + \frac{1}{p_{X}}} \cdot \frac{N}{N - 1 + \frac{1}{p_{Y}}},
\end{split}
\end{equation}
where $p_X = \mathbb{P}\left( \frac{p(\hat{X})}{M(p, \hat{p}) \hat{p}(\hat{X})}\mid \hat{X} = \hat{Y}\right) > 0$ 
and $p_Y = \mathbb{P}\left(\frac{q(\hat{Y})}{M(q, \hat{q}) \hat{q}(\hat{X})}\mid \hat{X} = \hat{Y}\right) > 0$, for $(\hat{X}, \hat{Y}) \sim \hat{\Gamma}$. Furthermore, using the maximality property of $\hat{\Gamma}$, the density of $\hat{X}$ conditionally on $\hat{X} = \hat{Y}$ is $\min(\hat{p}, \hat{q})$~\citep[][Lemma 2]{Wang2021maximal}, so that $p_X \geq \frac{1}{\mathbb{P}(\hat{X} = \hat{Y}) \, M(p, \hat{p})} \mathbb{E}\left[p(\hat{Y})\right]$, and similarly, $p_Y \geq \frac{1}{\mathbb{P}(\hat{X} = \hat{Y}) \, M(q, \hat{q})} \mathbb{E}\left[q(\hat{X})\right]$, with the identity $p_X = \frac{1}{M(p, \hat{p})}$ and $p_Y = \frac{1}{\mathbb{P}(M(q, \hat{q})}$ when we also have $\hat{p} = \hat{q}$.
Now we have
\begin{equation}
\begin{split}
    &\mathbb{P}(I = J) \\
    &= \mathbb{E}\left[\sum_{n=1}^N \min(W^X_n, W^Y_n)\right] \\ 
    &= \mathbb{E}\left[\sum_{n=1}^N \min\left(
        \frac{\frac{p(\hat{X}_n)}{\hat{p}(\hat{X}_n)}}
             {\sum_{m=1}^N \frac{p(\hat{X}_m)}{\hat{p}(\hat{X}_m)}},
        \frac{\frac{q(\hat{Y}_n)}{\hat{q}(\hat{Y}_n)}}
             {\sum_{m=1}^N \frac{q(\hat{Y}_m)}{\hat{q}(\hat{Y}_m)}}\right)
        \right]\\
    &= N\mathbb{E}\left[\min\left(
        \frac{\frac{p(\hat{X}_1)}{\hat{p}(\hat{X}_1)}}
             {\sum_{m=1}^N \frac{p(\hat{X}_m)}{\hat{p}(\hat{X}_m)}},
        \frac{\frac{q(\hat{Y}_1)}{\hat{q}(\hat{Y}_1)}}
             {\sum_{m=1}^N \frac{q(\hat{Y}_m)}{\hat{q}(\hat{Y}_m)}}\right)
        \right]\\
    &= \frac{N}{\sqrt{2 N \log \log N}}\mathbb{E}\left[\min\left(
        \frac{\frac{p(\hat{X}_1)}{\hat{p}(\hat{X}_1)}}
             {\frac{1}{\sqrt{2 N \log \log N}}\sum_{m=1}^N \frac{p(\hat{X}_m)}{\hat{p}(\hat{X}_m)}},
        \frac{\frac{q(\hat{Y}_1)}{\hat{q}(\hat{Y}_1)}}
             {\frac{1}{\sqrt{2 N \log \log N}}\sum_{m=1}^N \frac{q(\hat{Y}_m)}{\hat{q}(\hat{Y}_m)}}\right)
        \right] > 0,
\end{split}
\end{equation}
the positivity following from the same argument as for Proposition~\ref{prop:diagonal-rejection-coupling}.

We can notice that 
\begin{align}
    \frac{1}{\sqrt{2 N \log \log N}}\sum_{m=1}^N \frac{p(\hat{X}_m)}{\hat{p}(\hat{X}_m)} 
        & = \frac{1}{\sqrt{2 N \log \log N}}\sum_{m=1}^N \left[\frac{p(\hat{X}_m)}{\hat{p}(\hat{X}_m)} - 1\right] + \frac{N}{\sqrt{2 N \log \log N}}
\end{align}
so that 
\begin{align}
    \frac{1}{\sqrt{2 N \log \log N}}\sum_{m=1}^N \frac{p(\hat{X}_m)}{\hat{p}(\hat{X}_m)} \leq  \frac{N}{\sqrt{2 N \log \log N}} + \sup_{K \geq N} \frac{1}{\sqrt{2 K \log \log K}}\sum_{m=1}^K \left[\frac{p(\hat{X}_m)}{\hat{p}(\hat{X}_m)} - 1\right],
\end{align}
and
\begin{align}
    \frac{1}{\sqrt{2 N \log \log N}}\sum_{m=1}^N \frac{p(\hat{X}_m)}{\hat{p}(\hat{X}_m)} \geq  \frac{N}{\sqrt{2 N \log \log N}} + \inf_{K \geq N} \frac{1}{\sqrt{2 K \log \log K}}\sum_{m=1}^K \left[\frac{p(\hat{X}_m)}{\hat{p}(\hat{X}_m)} - 1\right].
\end{align}
Furthermore, the law of the iterated logarithm \citep[see, e.g.,][]{Shiryaev:1996} ensures that 
\begin{align}
    \limsup_{N \to \infty} \frac{1}{\sqrt{2 N \log \log N}}\sum_{m=1}^N \left[\frac{p(\hat{X}_m)}{\hat{p}(\hat{X}_m)} - 1\right] &= \mathbb{V}\left[\frac{p(\hat{X})}{\hat{p}(\hat{X})}\right]^{1/2}
\end{align}
and 
\begin{align}
    \liminf_{N \to \infty} \frac{1}{\sqrt{2 N \log \log N}}\sum_{m=1}^N \left[\frac{p(\hat{X}_m)}{\hat{p}(\hat{X}_m)} - 1\right] &= -\mathbb{V}\left[\frac{p(\hat{X})}{\hat{p}(\hat{X})}\right]^{1/2},
\end{align}
a similar result being true for the $\hat{Y}$ part.
By rearranging the inequalities together, this ensures that there exist two sequences $l_N$ and $u_N$ such that for any $N$ we have
\begin{align}
    l_N \leq \mathbb{P}(I = J) \leq u_N,
\end{align}
with 
\begin{align}
    l_N \sim \frac{\int \min(p(x), q(x))\dd{x}}{1 + u \, \frac{\sqrt{2 N \log \log N}}{N}}, \quad u_N &\sim \frac{\int \min(p(x), q(x))\dd{x}}{1 - u \, \frac{\sqrt{2 N \log \log N}}{N}},
\end{align}
where $u = \max\left(\mathbb{V}\left[\frac{p(\hat{X})}{\hat{p}(\hat{X})}\right]^{1/2}, \mathbb{V}\left[\frac{q(\hat{Y})}{\hat{q}(\hat{Y})}\right]^{1/2}\right)$.

\newpage
\section{Proof of Proposition~\ref{prop:proba}}
\label{app:proof-lower-bound}

Let $\hat{X}, \hat{Y}$ be the output of the reflection-maximal coupling of Algorithm~\ref{alg:reflection-coupling}.
Because it is a maximal coupling~\citep[][Lemma 2]{Wang2021maximal}, the reflection-maximal coupling is successful with probability
\begin{align}
        \mathbb{P}(\hat{X} = \hat{Y}) = 1 - \norm{\mathcal{N}(\mu_p, \hat{\Sigma}) - \mathcal{N}(\mu_q, \hat{\Sigma})}_{TV} = 2 F\left(-\frac{1}{2}\norm{Q^{-1/2}(\mu_p - \mu_q)}_2\right),
\end{align}
where $F$ is the cumulative distribution function of a standard Gaussian and we have
\begin{equation}
\begin{split}
    \mathbb{P}(\hat{X} \in A \mid \hat{X} = \hat{Y})
    &= \mathbb{P}(\hat{Y} \in A \mid \hat{X} = \hat{Y}) \\
    &= \frac{1}{\mathbb{P}(\hat{X} = \hat{Y})}\int \mathbbm{1}(x \in A) \min\left(\mathcal{N}(x; \mu_p, \hat{\Sigma}), \mathcal{N}(x; \mu_q, \hat{\Sigma})\right) \dd{x}.
\end{split}
\end{equation}

For the coupling of Algorithm~\ref{alg:rejection-coupling} to be successful, we need the dominating coupling to be successful and for both $\hat{X}$ or $\hat{Y}$ to be accepted, that is,
\begin{align}
    \mathbb{P}(X = Y) \geq \mathbb{P}(\hat{X} = \hat{Y}) \mathbb{P}(A_X = A_Y = 1 \mid \hat{X} = \hat{Y}).
\end{align}
Now, using generic distributions notations,
\begin{equation}
\begin{split}
    &\mathbb{P}(A_X = A_Y = 1 \mid \hat{X}) \\
    &= \frac{1}{\mathbb{P}(\hat{X} = \hat{Y})} \int \min\left(\frac{p(x)}{M(p, \hat{p}) \hat{p}(x)}, \frac{q(x)}{M(q, \hat{q}) \hat{q}(x)}\right) \min(\hat{p}(x), \hat{q}(x)) \dd{x} \\
    &\leq \frac{1}{\mathbb{P}(\hat{X} = \hat{Y})} \int \frac{p(x)}{M(p, \hat{p}) \hat{p}(x)} \cdot \frac{q(x)}{M(q, \hat{q}) \hat{q}(x)} \min(\hat{p}(x), \hat{q}(x)) \dd{x}, \label{eq:base-eq-gauss}
\end{split}
\end{equation}
because both terms in the minimum are at most $1$.
\begin{remark}
    The last line would have corresponded to using a different uniform sample $U$ to accept $\hat{X}$ and $\hat{Y}$ in Algorithm~\ref{alg:rejection-coupling}.
\end{remark}
We can now compute the following integral explicitly:
\begin{equation}
\begin{split}
    &\frac{p(x)}{M(p, \hat{p}) \hat{p}(x)} \cdot \frac{q(x)}{M(q, \hat{q}) \hat{q}(x)}\\
    &= \exp\left(-\frac{1}{2}\left[(x - \mu_p)^{\top}\left(\Sigma_p^{-1} - \hat{\Sigma}^{-1}\right)(x - \mu_p) + (x - \mu_q)^{\top}\left(\Sigma_q^{-1} - \hat{\Sigma}^{-1}\right) (x - \mu_q)\right]\right).
\end{split}
\end{equation}
Now, we have
\begin{equation}
\begin{split}
    &\quad \mathcal{N}(x; \mu_p, \hat{\Sigma}) < \mathcal{N}(x; \mu_q, \hat{\Sigma}) \\
    &\Leftrightarrow x^\top \hat{\Sigma}^{-1}(\mu_p - \mu_q) < \frac{1}{2}\left[\mu_p^\top \hat{\Sigma}^{-1} \mu_p - \mu_q^\top \hat{\Sigma}^{-1} \mu_q\right].
\end{split}
\end{equation}
Moreover, when we have $\mathcal{N}(x; \mu_p, \hat{\Sigma}) < \mathcal{N}(x; \mu_q, \hat{\Sigma})$, the integrand in \eqref{eq:base-eq-gauss} becomes 
\begin{equation}
\begin{split}
    &\quad \frac{p(x)}{M(p, \hat{p}) \hat{p}(x)} \cdot \frac{q(x)}{M(q, \hat{q}) \hat{q}(x)} \mathcal{N}(x; \mu_p, \hat{\Sigma}) \\
    &= \frac{1}{\det(2\pi \hat{\Sigma})^{1/2}} \exp\left(-\frac{1}{2}\left[(x - \mu_p)^{\top}\Sigma_p^{-1}(x - \mu_p) + (x - \mu_q)^{\top}\left(\Sigma_q^{-1} - \hat{\Sigma}^{-1}\right) (x - \mu_q)\right]\right).
\end{split}
\end{equation}
Now,
\begin{equation}
\begin{split}
    &(x - \mu_p)^{\top}\Sigma_p^{-1}(x - \mu_p) + (x - \mu_q)^{\top}\left(\Sigma_q^{-1} - \hat{\Sigma}^{-1}\right) (x - \mu_q) \\
    &= (x - \alpha)^\top H^{-1} (x - \alpha) + \beta,
\end{split}
\end{equation}
where 
\begin{equation}
\begin{split}
    H &= \left(\Sigma_p^{-1}+ \Sigma_q^{-1} - \hat{\Sigma}^{-1}\right)^{-1}\\
    \alpha &= H\, \left(\Sigma_p^{-1}\, \mu_p + \left(\Sigma_q^{-1} - \hat{\Sigma}^{-1}\right)\, \mu_q\right)\\
    \beta &= \mu_p^{\top} \Sigma_p^{-1} \mu_p + \mu_q^{\top} \left(\Sigma_q^{-1} - \hat{\Sigma}^{-1}\right) \mu_q - \alpha^{\top} H^{-1} \alpha.
\end{split}
\end{equation}

This means that when $\mathcal{N}(x; \mu_p, \hat{\Sigma}) < \mathcal{N}(x; \mu_q, \hat{\Sigma})$, the integrand in \eqref{eq:base-eq-gauss} becomes
\begin{align}
    \exp(-\beta / 2)\frac{\det(2\pi H)^{1/2}}{\det(2\pi \hat{\Sigma})^{1/2}} \, \mathcal{N}(x; \alpha, H).
\end{align}
Under $\mathcal{N}\left(\alpha, H \right)$, the event
\begin{equation}
\begin{split}
    &\quad              x^\top \hat{\Sigma}^{-1}(\mu_p - \mu_q) < \frac{1}{2}\left[\mu_p^{\top} \hat{\Sigma}^{-1} \mu_p - \mu_q^\top \hat{\Sigma}^{-1} \mu_q\right]\\
    &\Leftrightarrow    (x - \alpha)^\top H^{-\top/2} H^{\top/2}\hat{\Sigma}^{-1}(\mu_p - \mu_q) < \frac{1}{2}\left[\mu_p^\top \hat{\Sigma}^{-1} \mu_p - \mu_q^\top \hat{\Sigma}^{-1} \mu_q - 2 \alpha^{\top} \hat{\Sigma}^{-1}(\mu_p-\mu_q)\right]
\end{split}
\end{equation}
has probability
\begin{align}
    \mathcal{F}(\alpha, H) \coloneqq F\left(\frac{1}{2}\frac{\mu_p^{\top} \hat{\Sigma}^{-1} \mu_p - \mu_q^\top \hat{\Sigma}^{-1} \mu - 2 \alpha^{\top} \hat{\Sigma}^{-1}(\mu_p - \mu_q)}{\norm{H^{\top/2}\hat{\Sigma}^{-1}(\mu_p - \mu_q)}_2}\right),
\end{align}
where $F$ is the cumulative distribution of a standard normal distribution.
Symmetrically, when $\mathcal{N}(x; \mu_p, \hat{\Sigma}) > \mathcal{N}(x; \mu_q, \hat{\Sigma})$, the integrand in \eqref{eq:base-eq-gauss} becomes 
\begin{align}
    \exp(-\gamma / 2) \frac{\det(2\pi H)^{1/2}}{\det(2\pi \hat{\Sigma})^{1/2}} \, \mathcal{N}(x; \delta, H),
\end{align}
where 
\begin{equation}
\begin{split}
    \delta &= H\, \left(\Sigma_q^{-1}\, \mu_q + \left(\Sigma_p^{-1} - \hat{\Sigma}^{-1}\right)\, \mu_p \right), \\
    \gamma &= \mu_q^{\top} \Sigma_q^{-1} \mu_q + \mu_p^{\top} \left(\Sigma_p^{-1} - \hat{\Sigma}^{-1}\right) \mu_p - \delta^{\top} H^{-1} \delta.
\end{split}
\end{equation}
and the event $\mathcal{N}(x; \mu_p, \hat{\Sigma}) > \mathcal{N}(x; \mu_q, \hat{\Sigma})$ has probability $1 - \mathcal{F}(\delta, H)$ under $\mathcal{N}\left(\delta, H \right)$. Therefore, we finally obtain
\begin{align}
    \mathbb{P}(X=Y) \geq \frac{\det(2\pi H)^{1/2}}{\det(2\pi \hat{\Sigma})^{1/2}} \left[\exp(-\beta/2) \mathcal{F}(\alpha, H) + \exp(-\gamma/2) \left[1 - \mathcal{F}(\delta, H)\right]\right].
\end{align}

\newpage
\section{Reflection maximal coupling}
\label{app:reflection}
\begin{algorithm}[!htb]
    \caption{Reflection-maximal coupling}
    \label{alg:reflection-coupling}
    \DontPrintSemicolon
    \Fn{\RC{$a$, $b$, $\Sigma$}}{
        $z = \Sigma^{-1/2}(a - b)$\;
        $e = z / \norm{z}$\;
        Sample $\dot{X} \sim \mathcal{N}(0, I)$ and $U \sim \mathcal{U}(0, 1)$\;
        \uIf{
            $\mathcal{N}(\dot{X}; 0, I) \, U < \mathcal{N}(\dot{X} + z; 0, I)$\;
        }{
            Set $\dot{Y} = \dot{X} + z$\;
        }
        \uElse{
            Set $\dot{Y} = \dot{X} - 2 \left\langle e, \dot{X} \right\rangle \, e$\;
        }
        Set $X = a + \Sigma^{1/2}\, \dot{X}$ and $Y = b + \Sigma^{1/2}\, \dot{Y}$\;
        \Ret{$X$, $Y$}
    }
\end{algorithm}

\newpage
\section{Proof of Proposition~\ref{prop:explicit_cov}}
\label{app:optimal_cov}
Consider the problem in \eqref{eq:opt_cov}. Following \citet{vandenberghe1998determinant}, for any invertible matrix $C$, maximising $\log \det(\hat{\Sigma}^{-1})$ is equivalent to maximising $\log \det(C^\top \, \hat{\Sigma}^{-1} C)$ and $\hat{\Sigma}^{-1} \preceq \Sigma_q^{-1} \Leftrightarrow C^\top \, \hat{\Sigma}^{-1} C \preceq C^\top \, \Sigma_q^{-1} C$, and similarly for the other constraints. 

We can now consider $C = \Sigma_q^{1/2}$ so that solving \eqref{eq:opt_cov} is equivalent to solving 
\begin{equation}
    \label{eq:aux_opt_cov}
    \begin{split}
    &\max \log \det(X) \\
            & X \preceq I, \quad
                X \preceq Y, \quad
                X \succeq 0,
    \end{split}
\end{equation}
for $X = C^\top \, \hat{\Sigma}^{-1}\, C$ with $Y = C^\top \, \Sigma_p^{-1}\, C$. Now let $Y = V\, D \, V^\top$ be a diagonalisation of $Y$, with orthogonal matrix $V$. We can now use $V$ to transform \eqref{eq:aux_opt_cov} via $U = V^\top X \, V$ into the equivalent problem
\begin{equation}
    \begin{split}
    &\max \log \det(U) \\
            & U \preceq I, \quad
                U \preceq D, \quad
                U \succeq 0.
    \end{split}
\end{equation}
The solution to this problem is given by a diagonal matrix $\bar{U}$ with diagonal entries being $\bar{U}_{ii} = \min(1, D_{ii})$ for all $i$. This allows us to reconstruct a solution $\hat{\Sigma}_{opt}^{-1} = C^{-\top}\, V \, \bar{U} \, V^\top C^{-1}$, that is, $\hat{\Sigma}_{opt} = C\, V\, \bar{U}^{-1} \, V^\top C^\top$ to \eqref{eq:opt_cov}.

\newpage
\section{Illustration of Gaussian couplings}
\label{app:gauss}
\begin{figure}[!htb]
    \centering
    \resizebox{\textwidth}{!}{
    \input{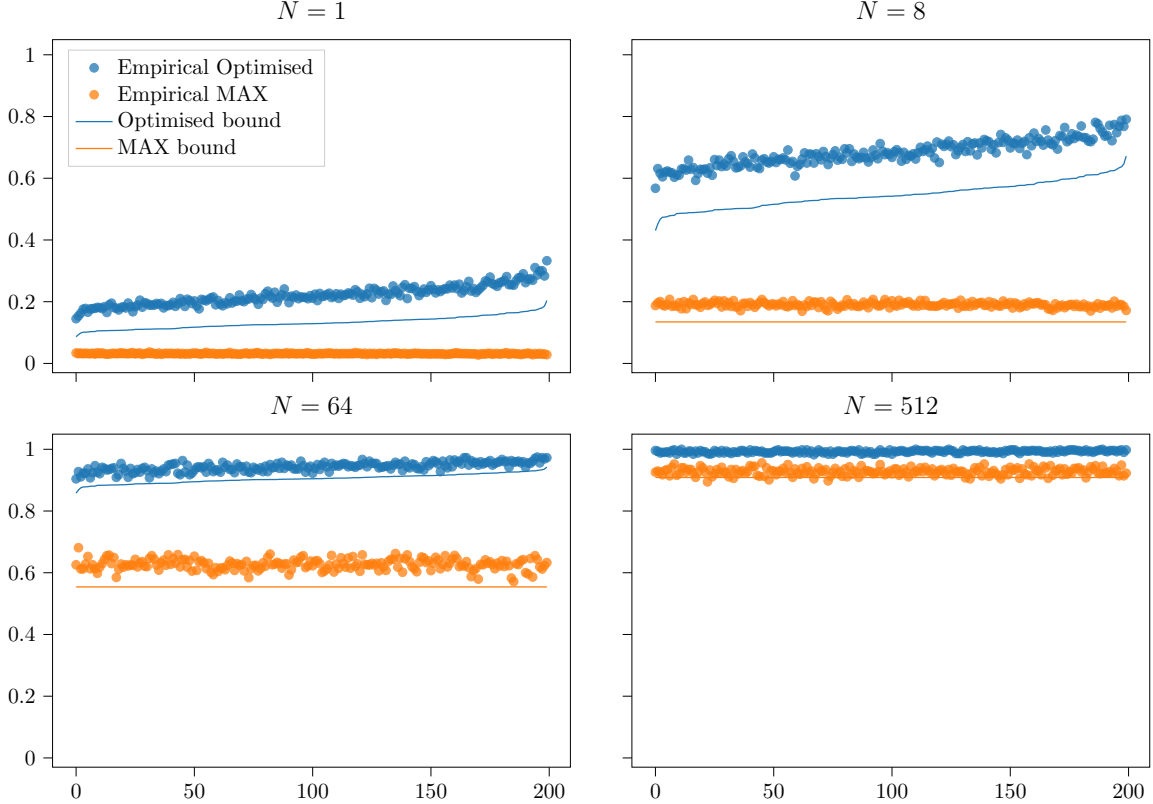}
    }
    \caption{Average number of steps $1 / \tau$ $N$ taken by Algorithm~\ref{alg:ensemble-rejection-coupling} to sample from the coupling using the reflection-maximal dominating coupling with $\hat{\Sigma}_{opt}$ and $\hat{\Sigma}_{max}$ and an ensemble proposal of size $N$. For better visualisation, the points have been ordered according to the value of the bound derived in Proposition~\ref{prop:rejection-coupling} using $\hat{\Sigma}_{opt}$ as the dominating covariance matrix.}
    \label{fig:acceptance_gaussian}
\end{figure}

To generate Figures~\ref{fig:acceptance_gaussian} and \ref{fig:coupling_gaussian}, we take the dimension $d$ for the underlying space $\mathbb{R}^d$ to be $10$, $\mu_p = \mu_q = 0$, $\Sigma_p = \mathrm{diag}\left(1, 2, \cdots, d\right)$ and we draw $200$ random $\Sigma_q$'s by sampling a random orthogonal matrix $U$ (given by the QR decomposition of a $d \times d$ matrix of i.i.d. standard Gaussian samples) and setting $\Sigma_q = U \Sigma_p U^{\top}$. We then run Algorithm~\ref{alg:coupled-rejection-resampling} with a reflection-maximal dominating coupling $500$ times for each pair of covariance matrices generated to compute the coupling and the acceptance probabilities.
    
\begin{figure}[!htb]
    \centering
    \resizebox{\textwidth}{!}{
    \input{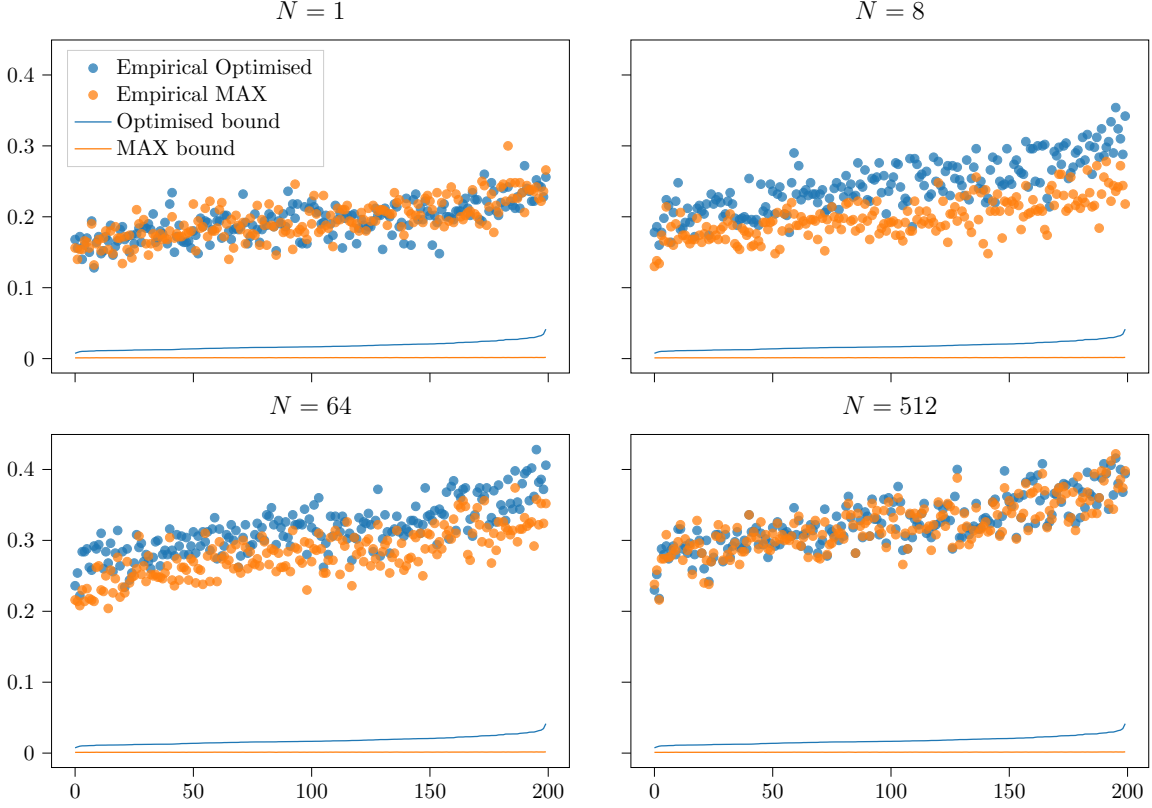}
    }
    \caption{Empirical coupling success probability of Algorithm~\ref{alg:ensemble-rejection-coupling} using the reflection-maximal dominating coupling with $\hat{\Sigma}_{opt}$ and $\hat{\Sigma}_{max}$ and an ensemble proposal of size $N$. For better visualisation, the points have been ordered according to the value of the bound derived in Proposition~\ref{prop:proba} using $\hat{\Sigma}_{opt}$ as the dominating covariance matrix.}
    \label{fig:coupling_gaussian}
\end{figure}
    
Figure~\ref{fig:acceptance_gaussian} shows that the number of trials required to sample from Algorithms~\ref{alg:rejection-coupling} and \ref{alg:ensemble-rejection-coupling} is positively impacted by using $\hat{\Sigma}_{opt}$ instead of $\hat{\Sigma}_{max}$ for all values of $N$. On the other hand, Figure~\ref{fig:coupling_gaussian} shows that optimising for $\hat{\Sigma}$ does not yield any benefit for $N=1$, but starts doing so for any $N > 1$. This is because, for $N=1$, the coupling success probability is driven mostly by that of the dominating coupling, while for $N>1$, the term $\max\left(\mathbb{V}\left[\frac{p(\hat{X})}{\hat{p}(\hat{X})}\right]^{1/2}, \mathbb{V}\left[\frac{q(\hat{Y})}{\hat{q}(\hat{Y})}\right]^{1/2}\right)$, appearing in Proposition~\ref{prop:asymptotics-bound}, drives the improvement and is directly impacted by $\det(\hat{\Sigma})$.

\newpage
\section{Maximal coupling of translated exponential distributions}
\label{app:max_coupling_expon}
In the experiment of Section~\ref{sec:gauss_tails} we use maximal coupling of translated exponential distributions as the proposal for the rejection coupling algorithm. The aim here is now to present details on how to sample from this maximal coupling. The method uses the following proposals:
\begin{equation}
\begin{split}
\hat{p}(x) &= \mathrm{Exp}(x; \mu,\alpha(\mu)) = \alpha(\mu) e^{-\alpha(\mu)(x - \mu)}, \qquad x \geq \mu, \\
\hat{q}(x) &= \mathrm{Exp}(x; \eta,\alpha(\eta)) = \alpha(\eta) e^{-\alpha(\eta)(x - \eta)}, \qquad x \geq \eta, 
\end{split}
\end{equation}
where $\alpha(z) = \frac{z + \sqrt{z^2 + 4}}{2}$. We assume that $\eta > \mu$ which also implies that $\alpha(\eta) > \alpha(\mu)$. 

Then we have
\begin{equation}
  c(x) = \frac{1}{Z^c} \min( \hat{p}(x), \hat{q}(x) ) = \begin{cases}
    \frac{1}{Z^c} \alpha(\mu) e^{-\alpha(\mu)(x - \mu)} ,& x \ge \eta \text{ and } x \le \gamma, \\
    \frac{1}{Z^c} \alpha(\eta) e^{-\alpha(\eta)(x - \eta)}, & x \ge \gamma, \\
    0 , & \text{otherwise},
  \end{cases}
  \label{eq:expt_c}
\end{equation}
where
\begin{equation}
  \gamma = \frac{\log \alpha(\eta) - \log \alpha(\mu) + \alpha(\eta) \eta - \alpha(\mu) \mu}{\alpha(\eta) - \alpha(\mu)}.
\end{equation}
We now get that
\begin{equation}
\begin{split}
  Z^c &= \int \min(p(x), q(x)) \, dx = 
   \exp(-\alpha(\mu) (\eta - \mu)) - \exp(-\alpha(\mu) (\gamma  - \mu)) + \exp(-\alpha(\eta) (\gamma - \eta)), \\
  Z^c_1 &= \int_{\eta}^{\gamma} \alpha(\mu) e^{-\alpha(\mu)(x - \mu)} \, dx = 
  \exp(-\alpha(\mu) (\eta - \mu)) - \exp(-\alpha(\mu) (\gamma  - \mu)), \\
  Z^c_2 &= \int_{\gamma}^{\infty} \alpha(\eta) e^{-\alpha(\eta)(x - \eta)} \, dx = 
  \exp(-\alpha(\eta) (\gamma - \eta)),
\end{split}
\end{equation}
where $Z^c$ is also the maximal coupling probability needed in the method. Both of the non-zero branches of \eqref{eq:expt_c} now have closed form inverse cumulative density functions (CDF):
\begin{equation}
\begin{split}
  C_1^{-1}(u) &= \mu - \frac{1}{\alpha(\mu)} \, \log\bigg[
    \exp(-\alpha(\mu) (\eta - \mu)) - u \, \left[ \exp(-\alpha(\mu) (\eta - \mu)) - \exp(-\alpha(\mu) (\gamma - \mu)) \right]
  \bigg], \\
  C_2^{-1}(u) &= \gamma - \frac{1}{\alpha(\eta)} \, \log (1 - u),
\end{split}
\end{equation}
and therefore we can sample from $c(x)$ by first drawing the branch with probabilities $Z^c_1 / Z^c$ and $Z^c_2 / Z^c$, and then by sampling with inverse CDF method by using the CDFs given above.

Let us then look at the density
\begin{equation}
\begin{split}
  \tilde{p}(x) &= \frac{1}{Z_p} \left[
  p(x) - \min( p(x), q(x) ) \right] \\
  &= \begin{cases}
    \frac{1}{Z^p} \left[ \alpha \, \exp(-\alpha(\mu) (x - \mu)) - \alpha(\eta) \, \exp(-\alpha(\eta) (x - \eta)) \right], & x \ge \gamma, \\
    \frac{1}{Z^p} \alpha(\mu) \, \exp(-\alpha(\mu) (x - \mu)), & \mu \le x \le \eta, \\
    0, & \text{ otherwise,}
  \end{cases} \\
\end{split}
\end{equation}
where the overall normalisation constant and the normalisation constants of the non-zero branches are given as
\begin{equation}
\begin{split}
  Z^p &= \int \left[ p(x) - \min( p(x), q(x) )  \right] \, dx \\
  &= 1 - \exp(-\alpha(\mu) (\eta - \mu)) + \exp(-\alpha(\mu) (\gamma  - \mu)) - \exp(-\alpha(\eta) (\gamma - \eta)), \\
  Z^p_1 &=  \int_{\gamma}^{\infty} \left[ \alpha(\mu) \, \exp(-\alpha(\mu) (x - \mu)) - \alpha(\eta) \, \exp(-\alpha(\eta) (x - \eta)) \right] \, dx \\
  &= \exp( -\alpha(\mu) (\gamma - \mu)) - \exp( -\alpha(\eta) (\gamma - \eta)), \\
  Z^p_2 &= \int_{\mu}^{\eta} \alpha(\mu) \, \exp(-\alpha(\mu) (x - \mu)) \, dx \\
  &= 1 - \exp(-\alpha(\mu) (\eta - \mu))
\end{split}
\end{equation}
We can now again first draw the branch with probabilities $Z^p_1 / Z^p$ and $Z^p_2 / Z^p$, and then use inverse CDF method in each of the branches. The second branch has a closed form inverse CDF
\begin{equation}
\begin{split}
  \tilde{P}^{-1}_2(u) &= \mu - \frac{1}{\alpha(\mu)} \log\bigg[
    1 - u \, \left[ 1 - \exp(-\alpha(\mu) (\eta - \mu)) \right]
  \bigg].
\end{split}
\end{equation}
However, the first branch does not have a closed form inverse CDF. Instead we need to solve the value of the inverse CDF $\tilde{P}^{-1}_1(u)$ at each $u$ by numerically solving the equation
\begin{equation}
\begin{split}
  \frac{1}{Z^p_1} \exp(-\alpha(\mu) (x - \mu)) - \frac{1}{Z^p_1} \exp(-\alpha(\eta) (x - \eta)) = u, \qquad x \ge \gamma.
\end{split}
\end{equation}
For that purpose, we use Chandrupatla's method \citep{Chandrupatla:1997}.

Finally, we need to sample from the distribution
\begin{equation}
\begin{split}
  \tilde{q}(x) &=
  \frac{1}{Z^q} \left[ q(x) - \min( p(x), q(x) ) \right] \\
  &= \begin{cases}
    \frac{1}{Z^q} \left[ \alpha(\eta) \, \exp(-\alpha(\eta) (x - \eta)) - \alpha(\mu) \, \exp(-\alpha (\mu) (x - \mu)) \right], & x \ge \eta \text{ and } x \le \gamma, \\
    0, & \text{ otherwise,}
  \end{cases} \\
\end{split}
\end{equation}
where the normalisation constant is
\begin{equation}
\begin{split}
Z^q = 1 - \exp(-\alpha(\mu) (\eta - \mu)) + \exp(-\alpha(\mu) (\gamma  - \mu)) - \exp(-\alpha(\eta) (\gamma - \eta)).
\end{split}
\end{equation}
We can again sample from $\tilde{q}(x)$ by inverse CDF method although the inverse CDF is not available in closed form. However, we can solve values of the inverse CDF $\tilde{Q}^{-1}(u)$ by numerically solving the equation
\begin{equation}
\begin{split}
  1 - \exp(-\alpha(\eta) (x- \eta)) - \exp(-\alpha(\mu) (\eta - \mu)) + \exp(-\alpha(\mu) (x - \mu)) = Z^q \, u,
\end{split}
\end{equation}
for $\eta \le x \le \gamma$. For this purpose we again used the Chandrupatla's method.

\newpage
\section{Coupled rejection resampling algorithm}
\label{app:resampling-algo}
In this section we reproduce the coupled (parallel) rejection resampling algorithm which is a modified version of the algorithm (Code 3) in \citet{Lawrence2016Parallel}. For the sake of concision, we only reproduce the non-ensemble version, the ensemble one being a direct modification following Algorithm~\ref{alg:ensemble-rejection-coupling}.

\begin{algorithm}[!htb]
    \caption{Coupled (parallel) rejection resampling of $w^X \in [0, \infty)^M$ and $w^Y \in [0, \infty)^M$}
    \label{alg:coupled-rejection-resampling}
    \DontPrintSemicolon
    \Fn{\RejR{$w^X$, $w^Y$, $\bar{w}^X$, $\bar{w}^Y$}}{
        \tcp{Supposing that for all $m \in \{1, \ldots, M\}$, $w^X_m \leq \bar{w}^X < \infty$ and $w^Y_m \leq \bar{w}^Y < \infty$.}
        \tcp{The loop below can be done in parallel}
        \For{$m=1, \ldots, M$}{
            Set $B^X_m = m$ and $B^Y_m = m$\;
            Set $A^X_m = 0$ and $A^Y_m = 0$ \tcp{Acceptance flags}
            Sample $U \sim \mathcal{U}(0, 1)$\;
            \uIf{
                    $U < w^X_{B^X_m} / \bar{w}^X$
                }{
                    Set $A^X_m = 1$\;
                }
            \uIf{
                    $U < w^Y_{B^Y_m} / \bar{w}^Y$
                }{
                    Set $A^Y_m = 1$\;
                }

            \While{
                $A^X_m = 0$ and $A^Y_m = 0$
            }
            {
                Sample $U \sim \mathcal{U}(0, 1)$\;
                Sample $B^X_m = B^Y_m$ from $\mathcal{U}(\{1, \ldots, N\})$\;
                \uIf{
                        $U <  w^X_{B^X_m} / \bar{w}^X$
                    }{
                        Set $A^X_m = 1$\;
                    }
                \uIf{
                        $U < w^Y_{B^Y_m} / \bar{w}^Y$
                    }{
                        Set $A^Y_m = 1$\;
                    }
            }
            \uIf{
                $A^X_m = 0$
            }{
                    Continue the rejection loop for $w^X$ only
            }
            \uIf{
                $A^Y_m = 0$
            }{
                Continue the rejection loop for $w^Y$ only
            }
            }
    \Ret{$B^X$, $B^Y$, $A^X$, $A^Y$}
    }
\end{algorithm}

\begin{remark}
    It is worth noting that the Metropolis-Hastings based parallel resampling method in \citet[][Code 2]{Lawrence2016Parallel} can also be coupled by coupling the individual chains appearing in the loop. This observation, although important, does not follow from the methods introduced in the current article, but is a direct consequence of the work of \citet{Jacob2020UnbiasedMCMC}. However, it is not directly clear if the resulting algorithm would target -- or even recover -- the maximal coupling between categorical distributions. Studying the behavior of the resulting resampling algorithm falls outside of the scope of this article.
\end{remark}

\newpage
\section{Thorisson algorithm}
\label{app:thorisson}
In this section we reproduce the modified version of Thorisson's algorithm~\citep{Thorisson2000Coupling,Gerber2020discussion}. This version of the algorithm introduces a parameter $0 < C < 1$ corresponding to the acceptable suboptimality of the coupling success. Formally, if one wants to sample $(X, Y)$ from a coupling of $(p, q)$, the resulting algorithm will have $\mathbb{P}(X=Y) = C \int \min(p, q)$. Introducing this variable enables control of the run time variance of the algorithm. $C=1$ corresponds to Thorisson's algorithm with infinite run time variance~\citep[see][for details]{Gerber2020discussion}. The algorithm is given in Algorithm~\ref{alg:thorisson}.

\begin{algorithm}[H]
    \caption{Modified Thorisson algorithm}
    \label{alg:thorisson}
    \DontPrintSemicolon
    \Fn{\Thorisson{$p$, $q$, $C$}}{
        Sample $X \sim p$\;
        Sample $U \sim  \mathcal{U}(0, 1)$\;
        \uIf{
            $U < \min(\frac{q(X)}{p(X)}, C)$
        }{
            Set $Y=X$\; 
        }
        \uElse{
            Set $A = 0$\;
            \While{$A \neq 1$}
            {
                Sample $U \sim  \mathcal{U}(0, 1)$\;
                Sample $Z \sim q$\;
                \uIf{
                    $U > \min\left(1, C \frac{p(Z)}{q(Z)}\right)$
                }{
                    Set $A=1$\; 
                }
            Set $Y=Z$\;
            }
        }
    \Ret{$X$, $Y$}
    }
\end{algorithm}

\newpage
\section{Run time efficiency of ensemble rejection sampling}
\label{app:efficiency}
We now provide a theoretical argument as to why using an ensemble proposal in the rejection sampler results in better overall run time when running on a parallel hardware such as a GPU. For the sake of simplicity, we make the following assumptions:

\begin{enumerate}
    \item $p_{\mathrm{RS}}$ is the acceptance probability of the standard rejection sampler with one proposal.
    \item The cost of sampling a proposal and computing the importance weight is given by a constant $K > 0$, and sampling $N$ proposals in parallel comes at the constant same cost $K$.
    \item The acceptance probability of the ensemble rejection sampler with $N$ particles is exactly (and not $\geq$) $p_{\mathrm{ERS}} = \frac{N p_{\mathrm{RS}}}{(N-1)p_{\mathrm{RS}} + 1}$.
    \item The run time of one step in the ERS algorithm is $K + \log_2(N)$, that is, the cost of sampling and computing weights $N$ times fully in parallel, and the cost of sampling one candidate using a prefix-sum type method~\citep[][Supplementary material]{Lawrence2016Parallel}. We note that this recovers the case of single proposal RS when $N=1$.
    \item The number of cores available is infinite.
\end{enumerate}

Let us understand the total run time of getting $M$ samples in parallel. Because parallelisation on a GPU is blocking~\footnote{For an informal discussion on the matter, in the context of rejection sampling, we refer the reader to \url{https://cas.ee.ic.ac.uk/people/dt10/research/rngs-gpu-rejection.html}, last accessed on 22/02/2022.}, the expected run time in the case of using simple RS consists in the expectation of the maximum of $M$ i.i.d. Geometric distributions with parameter $p_{\mathrm{RS}}$ times the cost of one step, and similarly for the ERS version with $p_{\mathrm{ERS}}$ instead. To fix notation, let us call $C(M, N)$ this quantity, and keep in mind that $N=1$ recovers the RS algorithm case. With these notations, the expected total run time of an algorithm sampling $M$ times in parallel with $N$ proposals per sample is therefore proportional to $[K + \log_2(N)] C(M, N)$.

\begin{figure}[!htb]
\begin{subfigure}{.48\linewidth}
\centering
\resizebox{\textwidth}{!}{
    \includegraphics{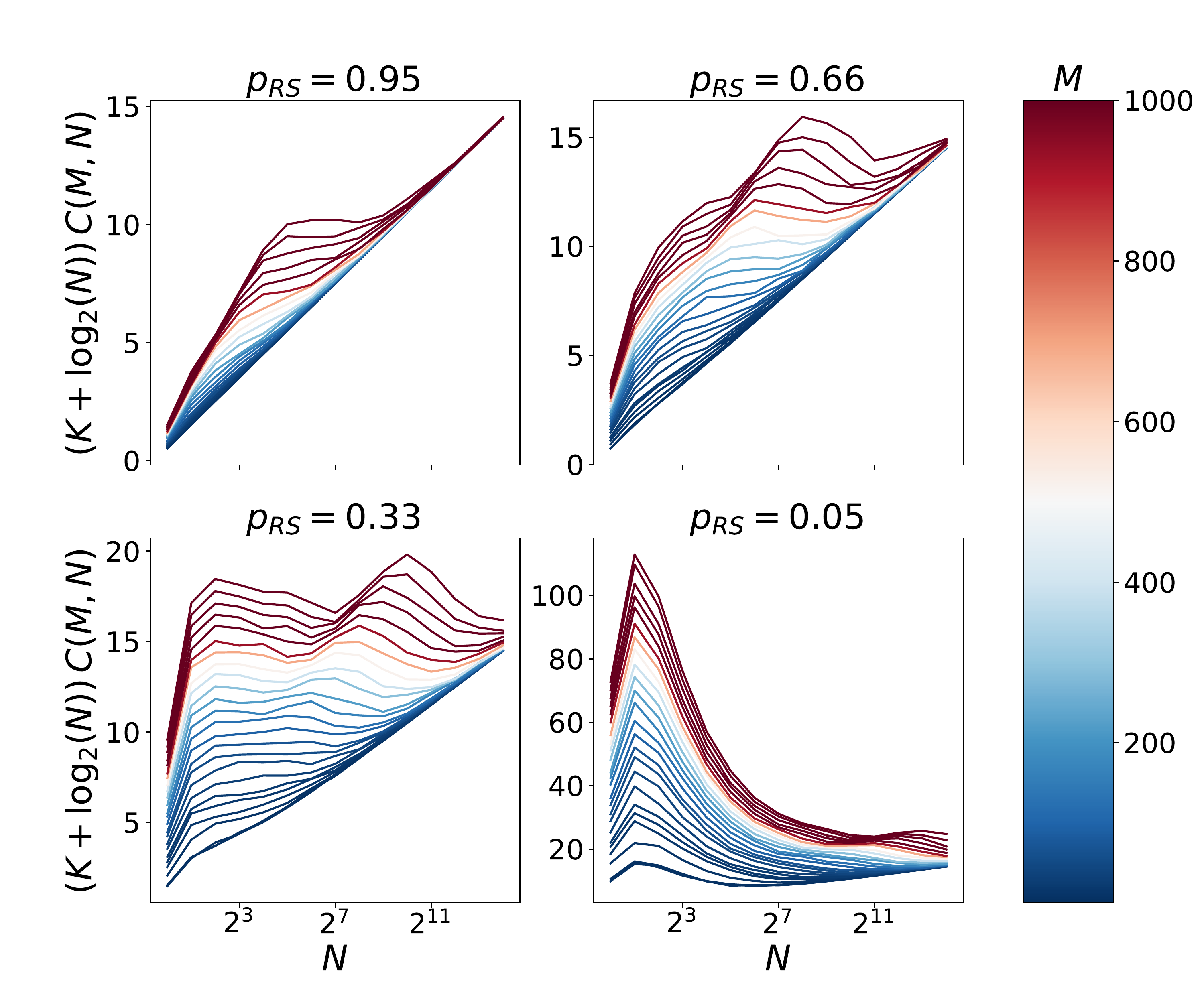}
}
\captionsetup{width=0.7\textwidth}
\caption{Cost analysis of the ensemble rejection sampler for $K=1/2$.}
\label{fig:cost_1/2}
\end{subfigure}%
\begin{subfigure}{.48\linewidth}
\centering
\resizebox{\textwidth}{!}{
    \includegraphics{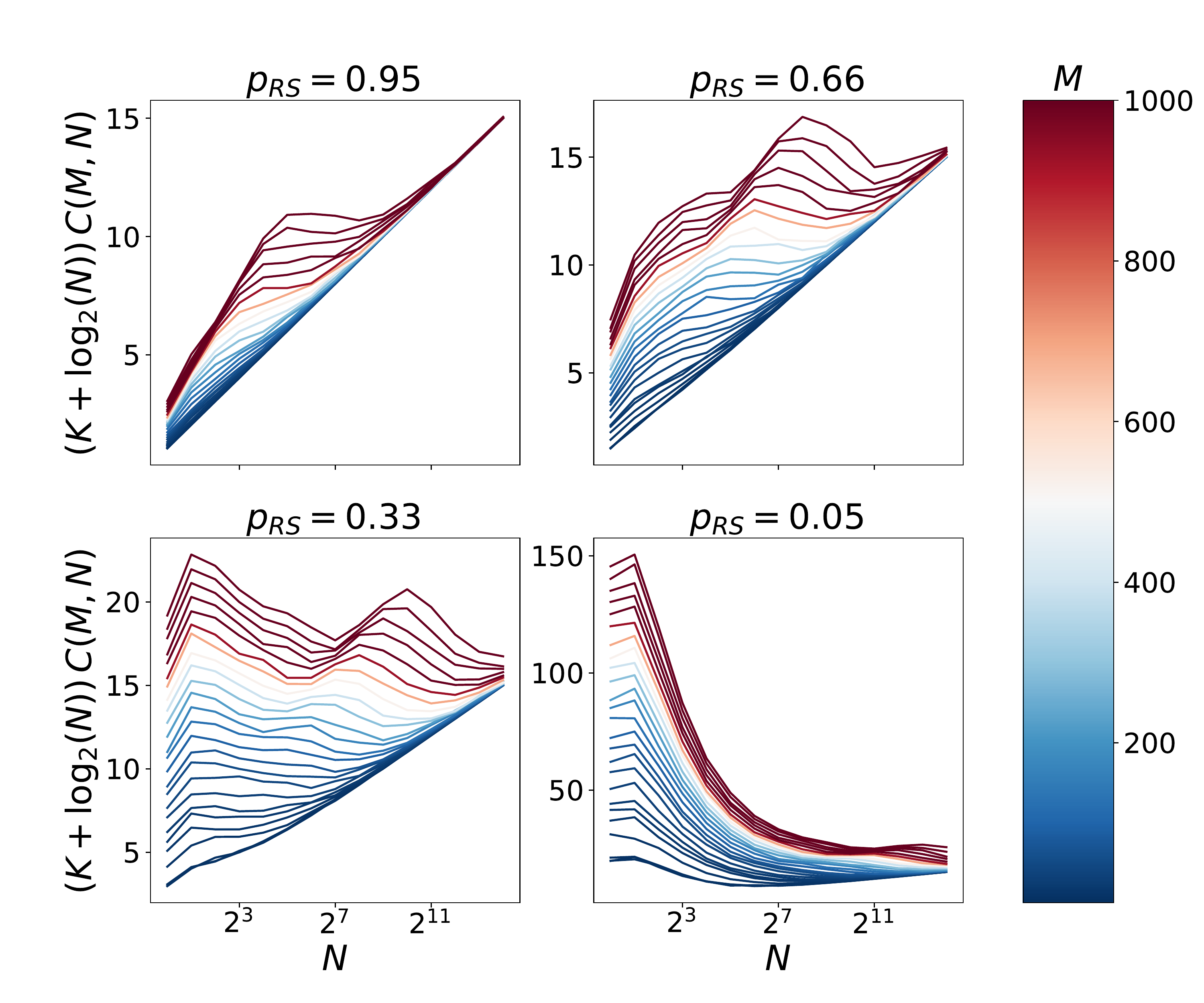}
}
\captionsetup{width=0.7\textwidth}
\caption{Cost analysis of the ensemble rejection sampler for $K=1$.}
\label{fig:cost_1}
\end{subfigure}\\[1ex]
\begin{subfigure}{\linewidth}
\centering
\resizebox{0.48\textwidth}{!}{
    \includegraphics{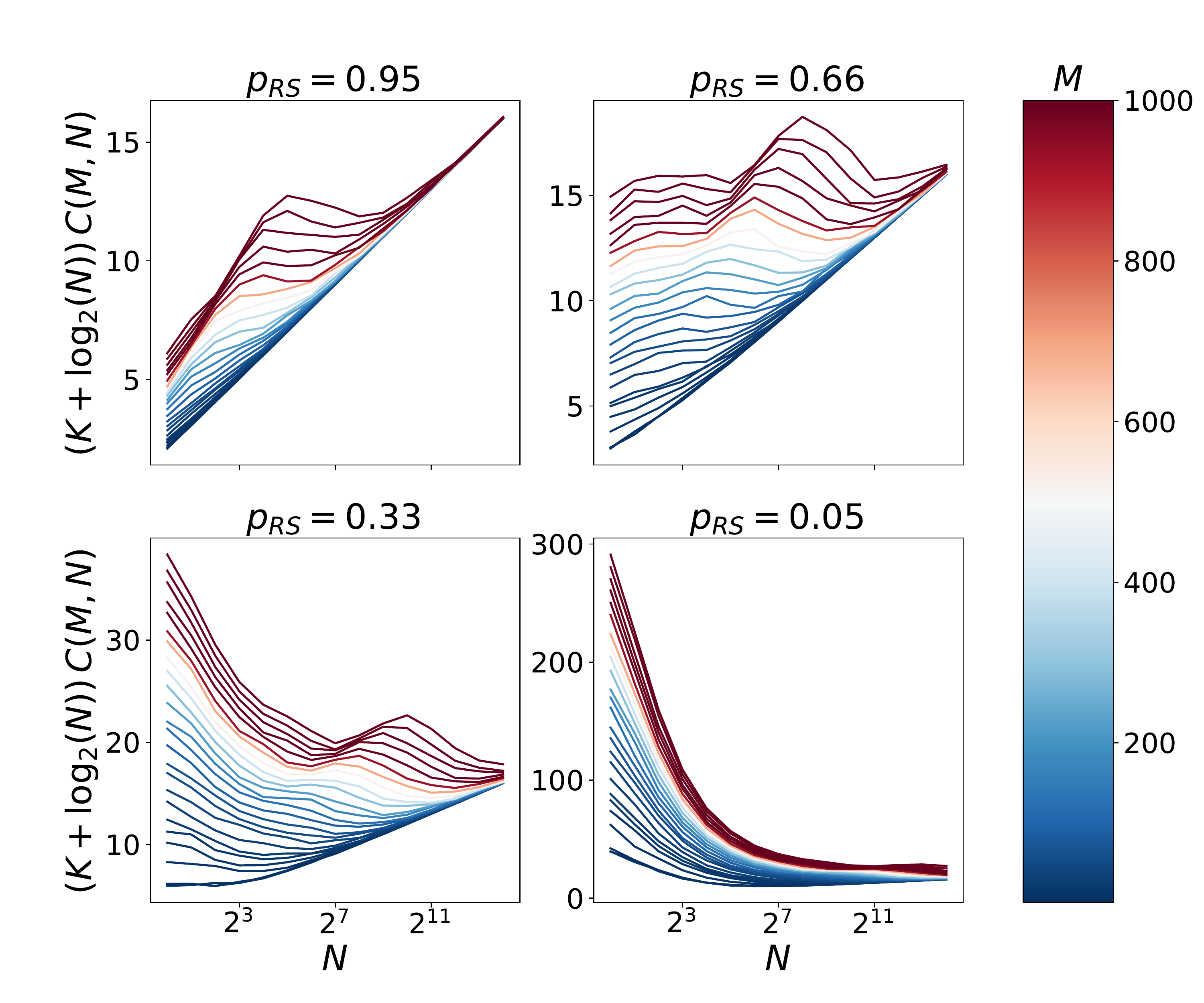}
}
\captionsetup{width=0.7\textwidth}
\caption{Cost analysis of the ensemble rejection sampler for $K=2$.}
\label{fig:cost_2}
\end{subfigure}
\end{figure}

In Figures~\ref{fig:cost_1/2}, \ref{fig:cost_1}, and \ref{fig:cost_2} we show the expected effective run time $N$, $M$, and $p_{\mathrm{RS}}$ for $K=1/2$, $K=1$, and $K=2$, respectively. We computed the (intractable) order statistics mean $C(M, N)$ using Monte Carlo.

As outlined by these, when either the cost of selecting a candidate in the ensemble is higher than the cost of sampling a single proposal ($K < 1$), or when the acceptance probability $p_{\mathrm{RS}}$ is large, using an ensemble proposal is detrimental to the total run time of the algorithm. However, as soon as the cost of sampling one proposal is large $K > 1$, and/or the acceptance probability $p_{\mathrm{RS}}$ is small, using an ensemble proves beneficial, with larger gains appearing when trying to sample from larger sample sizes $M$ from the target distribution. While analysing the actual parallel time complexity of the algorithm in the more realistic scenario of a finite number of cores is beyond the scope of the current paper, this analysis provides some intuition so as to why ensemble proposals may improve sampling speed on parallel hardware. This behavior is empirically confirmed in Section~\ref{subsec:resampling}.

\end{document}